\documentclass{article}
\usepackage{appendix}
\usepackage[a4paper, total={160mm, 247mm}]{geometry}
\usepackage[affil-it]{authblk}

\usepackage{cite}
\usepackage{amsmath,amssymb,amsfonts}
\usepackage{algorithmic}
\usepackage{graphicx}
\usepackage{textcomp}
\usepackage{bussproofs}
\usepackage{xcolor}
\usepackage{stmaryrd}
\usepackage{url}
\usepackage{todonotes}

\usepackage{amsthm}
\usepackage{xypic}

\newtheorem{lemma}{Lemma}
\newtheorem{corollary}{Corollary}
\newtheorem{theorem}{Theorem}
\newtheorem{proposition}{Proposition}

\theoremstyle{definition}
\newtheorem{definition}{Definition}
\newtheorem{example}{Example}

\newcommand{\cc}{{\cdot}}

\newcommand{\m}[1]{\mathcal{#1}}

\newcommand{\Ag}{\mathbb{A}}

\newcommand{\Tok}{\mathbb{T}}
\newcommand{\Mes}{\mathbb{M}}
\newcommand{\Form}{\mathbb{F}}

\newcommand{\X}{\mathtt{X}}
\newcommand{\tP}{\mathtt{P}}
\newcommand{\tF}{\mathsf{P}}

\newcommand{\ot}{\scalebox{0.5}[0.75]{$\leftarrow$}}

\newcommand{\RN}{\mathtt{N}}

\newcommand{\nval}[4]{\mathsf{val}_{#1} (#2 \ot #3) (#4)}

\newcommand{\mtrust}[4]{\mathtt{trust}_{#1} \, #2 , #3 \, (#4)}
\newcommand{\ntrust}[4]{\mathsf{trust}_{#1} (#2 \ot #3) (#4)}

\newcommand{\p}{A}

\newcommand{\I}{\mathcal{I}}
\newcommand{\J}{\mathcal{J}}

\newcommand{\F}{\mathcal{F}}
\newcommand{\G}{\mathcal{G}}
\newcommand{\B}{\mathcal{B}}
\newcommand{\D}{\mathcal{D}}

\newcommand{\T}{\mathcal{T}}
\newcommand{\R}{\mathcal{R}}
\newcommand{\M}{\mathcal{M}}
\newcommand{\N}{\mathcal{N}}
\newcommand{\E}{\mathcal{E}}

\newcommand{\Var}[1]{\texttt{var}(#1)}

\newcommand{\Lam}[2]{\lambda #1.(#2)}
\newcommand{\App}[2]{#1 \cdot #2}

\newcommand{\tLock}[2]{\texttt{lock}_{#1}(#2)}
\newcommand{\tKey}[3]{\texttt{key}_{#1 \Rrightarrow #2}(#3)}

\newcommand{\ip}{\rotatebox[origin=c]{180}{$\pi$}}

\newcommand{\valu}[1]{\widehat{\I_{#1}}}
\newcommand{\valm}[1]{\widehat{#1}}
\newcommand{\valw}[1]{\widehat{\J_{#1}}}

\def\BibTeX{{\rm B\kern-.05em{\sc i\kern-.025em b}\kern-.08em
    T\kern-.1667em\lower.7ex\hbox{E}\kern-.125emX}}
\begin{document}

\title{Modal Logic for Distributed Trust\thanks{This research has been supported by Estonian Research Council, grant No. PRG1780.}
}

\author{Niels Voorneveld \qquad \qquad Peeter Laud}
\affil{Cybernetica AS}

\maketitle

\begin{abstract}
	We propose a method for reasoning about trust in multi-agent systems, specifying a language for describing communication protocols and making trust assumptions and derivations. This is given an interpretation in a modal logic for describing the beliefs and communications of agents in a network. 
	We define how information in the network can be shared via forwarding, and how trust between agents can be generalized to trust across networks.
	
	We give specifications for the modal logic which can be readily adapted into a lambda calculus of proofs.
	We show that by nesting modalities, we can describe chains of communication between agents, and establish suitable notions of trust for such chains.
	We see how this can be applied to trust models in public key infrastructures, as well as other interaction protocols in distributed systems.

\end{abstract}


\section{Introduction}\label{sec:intro}

When protocols are distributed across multiple parties, you must consider those parties in their proper context in order to make guarantees. Who has access to what information? Who is claiming that the information is correct. These are factors necessary for determining both whether something is correct, and whether someone has the necessary knowledge to learn about this correctness.
Properties and guarantees we want to establish within distributed contexts involve claims regarding credentials, and correctness of distributed processes according to some specifications.
Cooperation is needed to prove such guarantees, which fundamentally requires the sharing of statements and the consideration of other peoples statements. 

There are two sides to cooperation. The first is to facilitate the building of strong claims by combining knowledge. This may not necessitate full transparency, and entities can choose what and with whom they share their claims, in order to cut costs or keep confidential information hidden. Such entities would like to learn the effectivity of sharing facts, and what others can do with them. They can prove that they can convince others, or that collaborators can use their facts effectively.
Though the partial \emph{hiding} of information is assumed, this paper is not about \emph{privacy}. No guarantees are made that information cannot leak. It is simply specified in which contexts claims are officially shared, and proofs are built only on such specified assumptions. Still, some basic properties of privacy across communication channels can be asserted, and one can reason about consequences of leaked information.

The flipside to cooperation is \emph{responsibility} and \emph{accountability}. Sharing information is based on mutual trust, and entities making false claims can break this trust. As such, entities should take care not to share false facts, and careful entities should stick to their expertise. 
Entities are not just accountable for the claims they make, but also for the logical consequences of their claims. Implied lies are still lies. As such, it is important to specify exactly what information an entity has access to. This is another reason why the destination of claims matters, as it tells us who has received the information. Entities should not be allowed to ignore information shared with them. They should consider consequences of received claims, and be accountable for them.
That said, it is important to distinguish between the claim as a whole, and the content of the claim. It is ok to accept that a claim has been made, yet to deny that the content of the claim is true. Whether or not to accept the contents of a claim as true is a matter of trust. 

Trust is the cornerstone of this work, and has been studied from many points of view. If you trust a claim, you accept its content as true. Trust is formulated in a conditional way: it does not matter whether the claim has actually been made, trust states that if a certain claim is made by a certain entity, then it is true. 
Making trust conditional facilitates cooperation. Entities may declare their trust in others. This can then be used by other entities, which could incorporate it in their model in order to make guarantees. For instance, a shop may declare trust in an authority in its ability to authenticate credentials. A buyer then knows it can prove their credentials to the shop via this authority. Such interacting trust relations can then be built up in order to prove guarantees of larger authentication networks and protocols.
Note however, that if an entity receives a claim which they have previously declared as trustworthy, they are now accountable for the consequences of the claim. Hence, \emph{declared trust}, which can be interpreted as \emph{vouching} for the expertise of someone, requires one to accept accountability for a trusted entity's statements regarding that area of expertise.

\subsection{Constructive Modal Logic}

Constructive logics have many useful applications in computer science and security information systems.
Such logics force one to be more specific regarding how a property is satisfied. Proving a disjunction requires one to specify explicitly which of the two options are true (or give a decision procedure as such given the assumptions). Moreover, it disallows double negation elimination: ``it cannot fail" is not a sufficient enough reason of why something is guaranteed to work.

In short, constructive logic requires one to build an explicit guarantee using a set of assumptions. This also increases accountability: the precise reason for the guarantee is embedded within the proof itself, allowing one to trace back which assumptions are used, and how a guarantee is satisfied. If the guarantee is later shown to be false, then it is easier to identify the entity that made a mistake.

To express both the claims being shared, and the internal perspectives of entities, we use modalities. Guided by our pursuit of accountable claims, we use modalities satisfying the axiom K and the necessity rule, which are considered standard axioms in epistemic (knowledge-based) logics. 
We moreover consider a collection of well-behaved axioms relating the many perspectives. Though many varieties can be considered, we start with a basic model of \emph{belief} $\B$ and \emph{interaction} $\I$ to mark internal and communicated knowledge.

A core building block of our logic is the consideration of modal \emph{splitting axioms}, which declare that $\M A \Rightarrow \N \R A$ holds for any statement $A$. These express how information $A$ from modality $\M$ can be adopted by modality $\N$. This latter domain does not necessarily accept the statement as true. Instead, they consider it as flagged by modality $\R$, and can reason with this statement. A specific instance of such an axiom is an \emph{awareness axiom} of the form $\M A \Rightarrow \N \M A$ for all $A$, where we say $\N$ is aware of all statements $A$ provable by $\M$.
This would hold for instance if $\M$ represents a communication channel which $\N$ has access to.

Another aspect to consider is the \emph{decidability} of the logic; does an algorithm exist which shows whether a guarantee is true given certain assumptions? Decidability is useful in order to increase accountability: it enables entities to compute and hence be aware of certain consequences of their claims. 
Checking for particularly significant consequences of one's claims should be part of an entity's due diligence.
In order to facilitate decidability, we limit ourselves to a specified set of modal axioms, as well as avoiding the use of polymorphism and quantifiers.

\subsection{Distributed Trust and Decentralized Identity}

One of the main motivations for this paper is its applications to decentralized identities.
We believe that the emergence of \emph{self-sovereign identity}~\cite{Reed-building-blocks-SSI} acutely necessitates the ability to assign truth values to complex statements about and including trust. An entity's sovereignty over their identity means the proliferation of identifiers that that entity intends to use in different contexts. These identifiers are presumably created by the entity themselves; some authority may or may not have bound them to identifiers in a more generally known namespace (e.g. the names of the residents of some country).

The statements that the various authorities make about an entity will be bound to the various identifiers created by that entity (as expressed by e.g. the ``Privacy and minimal disclosure'' principle of~\cite{sovrin-principles-of-SSI}). While it should generally be the responsibility of that entity to make sure that the claims they want to present together are bound to the same identifier, we expect the reality to not be so simple. When a relying party (i.e. a party that makes decisions on the basis of the claims it receives about an entity) receives a number of statements pertaining to a number of different identifiers, as well as some statements that relate these identifiers to each other (as per the ``Delegation'' principle of~\cite{sovrin-principles-of-SSI}), it becomes highly non-trivial to find out which of these claims can be combined with each other. This paper provides a system for relying parties to make such inferences.

The same question --- what can be reasonably and confidently derived for which identifier? --- is also something that the individual entities themselves want to find the answers to. Our logical approach provides them with tools to make these inferences \emph{from the point of view of some other party}, allowing them to forecast the decisions made by the relying parties. Indeed, the decisions can be derived for various parties, considering what we know about their trust relationships; these decisions may be combined with each other and considered from the point of view of different parties.

\subsection{Misplaced Trust in Certificate Authorities}

A motivating example we focus on is that of networks for public key validation.
A certificate authority (CA) is an entity in such networks which makes publicly verifiable statements that bind public keys to entities. These statements are the \emph{certificates}. A CA is trusted to be able to properly verify that an entity is in control of a public key. We consult CAs to find out an entity's public key. There are various mechanisms to find out whether an entity is a CA, e.g. a higher-level CA may issue a certificate to another entity stating that it is a (sub-)CA, too. We use such statements in hierarchical public-key infrastructures (PKI), where the knowledge of the public key of a \emph{root} CA allows us to verify a \emph{certificate chain} that ultimately binds a public key to an (end-)entity.

Of course, when an entity is declared to be a CA, then the declarer has to perform due diligence on that entity's ability to be a CA. Sometimes these declarations are made by CAs that should not be declaring other entities to be a CA. This creates confusion when verifying certificate chains, because the last certificate in that chain may be issued by a CA that should not be seen as a CA. Several incidents (see~\cite{spof23report}, subsections A.1.11, A.1.13, A.1.14, A.1.17) have taken place due to such misplaced trust. Even more (see subsections A.1.10 and A.1.25 in~\cite{spof23report}) have happened due to other usage restrictions that have been absent in the issued certificate.

We employ our logic as a tool for specifying the function and duties of entities in distributed networks, like CAs in public key infrastructures. This way, they can be held accountable for mistakes made. Given a formal constructive proof of a property, if the property is later shown to be false, we can trace back the dependencies and find who is at fault, ensuring accountability. Moreover, definitions of trust can be fine-tuned to safeguard against certain errors and attacks. For instance, threshold trust can be employed to deal with issues of unresponsiveness or corruption of CAs.

\subsection{Related Work}

There exist a number of calculi and logics for defining and computing trust in distributed networks in general and PKIs in particular. 
The paper \cite{Rangan88} considers a modal logic with belief modality, with the description of messages modeled by the underlying Kripke style semantics instead of explicitly described in the logic as we do.
Non-modal approaches include a predicate calculus \cite{Bakkali01} studying a generalization of the usual graph description of PKI networks, and \cite{Aziz07} which focuses on modeling the concurrent properties of such networks. 
Especially interesting are \cite{JingWei09,Singh2011TrustAD} which both study trust statements and inferences which can be made using them, similar to our treatment of trust for PKI, which we describe within our logic. 
The goals of these calculi are similar to our formalisms --- turning statements made by trusted entities into beliefs of users. Though they do not support \emph{higher-order} reasoning, with messages referring to other messages, which our modal logic does consider.

We use a modal logic to describe and give meaning to trust in distributed systems. Modal logics have had many applications in computer science due to their flexibility and versatility, which shows their usefulness as a verification tool. Examples include modeling programming features such as variable scoping \cite{Nanevski08}, and reasoning about cryptographic processes \cite{FRENDRUP2002124}. Using a modal logic to describe trust has been investigated in several different contexts. We ground our work on the definition of \emph{cautious trust} by Liau \cite{LIAU2003}, given in terms of belief and communication. This has been further investigated theoretically in a plethora of other works \cite{Dastani05,Dundua2010TrustAB,Leturc18}.
It has been shown that such trust definitions could be extended to consider probability \cite{Herzig03,Kohlas08} and time \cite{Aldini14}.

We specifically look at \emph{intuitionistic} modal logic \cite{Wolter1999}, in order to keep proofs constructive and expressible by various lambda calculi \cite{Bellin03,MDTT,Acclavio2023CanonicityIM} via the Curry-Howard correspondence. We maintain that calculi such as these are one of the requirements for a trust logic, due to the support it can give for the derivation of trust in software libraries.

\subsection{Comparison to Authorization and Access Logics}

The work has a strong connection to \emph{authorization and access logics} (AL) \cite{Authentication,Garg06,Becker07,Becker12}. Though the basic motivation and foundations of the formalisms are distinct, they are not wholly incompatible. ALs are about certain commands, and whether they were made by an authorized entity. This entails authenticating proofs of authority. In this paper, we do not focus on specific commands and control statements as such, but instead focus on the sharing and evaluating of information across distributed networks.

Many works mainly consider two primitives, a modality $S_a$ for what a principal $a$ \emph{says}, and a primitive formula stating that $a$ \emph{speaks for} $b$. 
This logical operation already existed in foundational work \cite{Authentication}, which focused on information regarding who owns which keys, and used in particular descriptions of what information people see, believe and control. Soon after, in \cite{Abadi93}, a more thorough investigation of the \emph{says} constructor was done, studying a lattice structure and defining \emph{control} as the logical affirmation: saying a statement makes that statement true.

Many different axiomatizations of $S_a$ were considered, with order structure, consistency $S_a \bot \Rightarrow \bot$ (axiom D) and idempotency $S_a A \Leftrightarrow S_a S_a A$ occurring in \cite{Abadi93} already.
Consistently throughout further work \cite{Abadi06, Modal_Access, GargInterference} the axioms $A \Rightarrow S_a A$ and $S_a S_a A \Rightarrow S_a A$ were assumed, with some (mainly \cite{Abadi06}) assuming even $S_a S_b A \Rightarrow S_b S_a A$.
It is perhaps this last axiom which highlights a fundamental feature of authorization logic different from out approach: in ALs they are concerned about who has authority, not about who said something.

The axiom $A \Rightarrow S_a A$ signifies that whoever is reasoning about statements from $a$ may complement any of $a$'s utterances with other things which are true. This comes from the perspective that $a$ may want access or authority over something, and hence is ok with others completing their proof of access with additional details which may help them. This is in stark opposition to the aim in this paper, where $a$ is accountable for their statements, and would not want others to put words into their mouth. Additionally, we would like $a$ to be able to reason about consequences of their claims, and hold a precise model of how other people regard them. As such, they should only need to be accountable for knowledge they are aware of.

The paper\cite{Abadi03} considers a variation in which $A \Rightarrow S_a A$ is replaced by $S_b A \Rightarrow S_a S_b A$, meaning one may adopt claims made by others. Utterances by $b$ are considered as universally shared, which would make $a$ aware of them. We would like to disallow this as well in general, as communications are not necessarily shared with everyone. Other work\cite{Hirsch2013,Hirsch2020,Hu10} allow $S_a A \Rightarrow S_a S_a A$ as we do when considering belief, though still have $S_a S_a A \Rightarrow S_a A$.

We drop axioms of the form $\M \N A \Rightarrow \R A$ mostly in order to keep decidability in the absence of S4's unit axioms. We predominantly focus on how information is shared and adopted, which entails axioms of the form $\M A \Rightarrow \N A$ and $\M A \Rightarrow \N \R A$. These include $\B_a A \Rightarrow \B_a \B_a A$, as well as agents adopting information they have access to, such as $\I_{a \ot b}A \Rightarrow \B_a\I_{a \ot b} A$. 
We moreover do not consider axiom D stating that $\M \bot \Rightarrow \bot$ \cite{Iranmanesg08,Sirer11}, since we do not wish the inconsistency of one party's statements to necessarily imply the inconsistency of the whole.

Certain foundational works consider lattice structures on their agents\cite{Abadi93,Abadi06} which give an axiomatic order on $S_a$, as well as meet and join axioms. We adopt the order structure to define how belief is shared, which allows us to model \emph{broadcasting}. Other works have this order as dynamic statements in the logic, with derivable \emph{speaks for} formulas. We do not consider those, as they are statements of authority, not necessarily trust.
Some work formulating \emph{can} statements\cite{Becker07,Becker12} does bear resemblance to our treatment of trust. There, if $a$ says that $b$ can do $A$, and if $b$ does it, then $a$ says $A$ happens. Similarly, we say that if $a$ claims that $b$'s opinion on $A$ is valid (trust), and if $b$ claims $A$, $a$ is responsible for claim $A$. 

The DKAL language \cite{Gurevich08,Gurevich09} considers a similar trust statement as primitive in a different context. In terms of treatment of trust itself, DKAL is perhaps the closest to our logic, with our derivations of validity (Sec.~\ref{sec:trust}) having similarities with their operational semantics~\cite{Gurevich09}. 
DKAL mostly focuses on parametric modalities satisfying axiom K and Necessity, and does not consider many further axioms, such as the connections between modalities we explore in this paper.
Our treatment of multi-step information sharing (Sec.~\ref{sec:chains} and~\ref{sec:forward}) is more fine-grained, with proper significance given to both the senders and the \emph{intended} recipients of messages. We are also able to handle \emph{disjunction}, which is useful to model certain protocols. But our greatest difference lies in semantics --- the Kripke semantics of our logic gives us greater confidence in our axioms, inference rules, and proof search methods. DKAL instead has a translation to Liberal Datalog~\cite{EFPL=LiberalDatalog} for deriving truth. Their follow-up DKAL~2~\cite{MSRTR-2009-11} does have a Kripke model, but lacks some of the expressivity discussed above.

We consider our logic to be a specialization of former work on trust management systems~\cite{JingWei09,Singh2011TrustAD,Aldini14} in the direction of distributed systems with defined information access. We focus in particular on how to specify protocols for sharing, forwarding, and trusting claims across networks. We do this by adopting and modifying the theory of belief and interaction from the past \cite{LIAU2003} and forming a decidable logic framework for specifying trust and interaction protocols.

\subsection{Contributions and Paper Outline}

Our first contribution is the formulation of a theory to reason about communications over distributed networks, enabling us to prove the existence of trust across chains of communications. We establish a set of axioms (Figure~\ref{fig:axioms}) for sharing and considering information, which have not been the focus of previous work, and discuss how to reason about indirect communication and trust in the logic. We also touch upon a calculus for building proofs in the logic.

The language for sharing and considering information allows us to state, what it means for one agent to trust another one with respect to the validity of a statement. Our calculus gives us some useful inference rules for validity in multi-agent systems (Figure~\ref{fig:validityproperties}).

We extend the trust with \emph{levels}, allowing us to model examples such as public key infrastructures (PKI) that may involve several steps of delegation. Our language will be rich enough to express these levels, and to derive certain properties of them (Theorem~\ref{lem:shared}) that turn to be rich enough to concisely describe the trust relations in networks of agents.

We start in Section~\ref{sec:modal} by discussing how we describe the relevant concepts with modalities and axioms. In Section \ref{sec:chains} we show different ways of formulating indirect communications. We discuss reasoning with trust in Section \ref{sec:trust}, and formulate a formalism for sharing information and trust in a network in Section \ref{sec:forward}. We give further examples in Sections~\ref{sub:exam} and~\ref{sec:threshold}.
We discuss two calculi in Section \ref{sec:calculi}, and a Kripke model in \ref{sec:model},
 and give final remarks in Section~\ref{sec:conclusions}.

\paragraph{Related papers:} Since the writing of this article, two papers have been published focusing on and expanding on particular parts of this article.
One paper~\cite{TAB} develops the decidability of a schema of intuitionistic modal logic, to which is partially covered here, and furthermore develops how to derive relevant consequences. Another paper~\cite{NordSec} covers how you can reason about trust thresholds.

\section{Modal Reasoning}\label{sec:modal}

We use a constructive modal logic to specify both the perspectives of different entities, as well as claims entities make to each other.
We consider the general concept of \emph{domain}, which subsumes specified entities, groups of entities, the public domain, and other entities identified by public keys. 
Let $\Ag$ be a finite set of entities, domains or agents.

Following the traditional literature on trust, we consider two modalities which we will give a specified meaning:
\begin{itemize}
	\item For each $a \in \Ag$, the modality $\B_a$ signifies what $a$ believes. In general the statement $\B_a A$ means that we can prove that given presumed assumptions of agent $a$, we can show that $a$ should accept $A$ as true. We can think of it as $a$ can verify $A$ given their own knowledge and expertise, and given what has been communicated to and shared with them.
	\item For each pair $a, b \in \Ag$, the modality $\I_{a \ot b}$ signifies interaction between $a$ and $b$. The statement $\I_{a \ot b} A$ expresses that through interaction between $a$ and $b$, $b$ has claimed and communicated to $a$ enough information (claims) to imply that statement $A$ is true.
\end{itemize}

There is a significant difference between what one believes and what one claims to be true. Firstly, in general one would not assume that $\I_{a \ot b} A \Rightarrow \B_b A$ as $b$ may be more liberal with the truth than what they themselves truly believe and act upon. On the other hand, $\B_b A \Rightarrow \I_{a \ot b}A$ is not necessarily true as $b$ may not want to communicate all that they know to $a$.

Entities are accountable for the logical consequences of their claims on top of the claims themselves. Similarly, we assume they can make logical inferences within their own model of belief. As such, we assume each modality $\M$ satisfies the axiom K and the ``necessity'' inference rule.

\begin{itemize}
	\item Axiom K states that $\M A \Rightarrow \M (A \Rightarrow B) \Rightarrow \M B$ for any formulas $A$ and $B$, signifying that if we learn that both $A$ and $A \Rightarrow B$ hold in belief or interaction expressed by $\M$, then the consequence $B$ should also hold there. This axiom allows us to reason about inference that can be made within $\M$ as an external observer.
	\item Necessity states that if formula $A$ is undeniably true (a logical tautology), then so is $\M A$. This means everyone should accept obviously true logical statements; those which can be proven without further assumptions. This is weaker than saying $A \Rightarrow \M A$, which states that if $A$ happens to be true (e.g. given as an optional assumption), then $\M A$ is true. The latter implies $\M$ knows all details of the current situation, whereas necessity only states $\M$ knows about facts which hold in all possible situations.
\end{itemize}

On top of the general axioms, we assume some additional ones which signify who has access to what information. We predominantly consider \emph{splitting} axioms of the form $\M X \Rightarrow \N \R X$.
Awareness axioms state that agents down the line can reflect on and consider domains they have knowledge about:
\begin{itemize}
	\item For any $A$, $\B_a A \Rightarrow \B_a\B_a A$.
	\item For any $A$, $\I_{a \ot b} A \Rightarrow \B_a\I_{a \ot b} A$.
\end{itemize}
So an agent is aware of their own internal viewpoint, as well as everything communicated to them.

Lastly we express the fact that claims made are intended as expressions of belief. For instance, if an agent $a$ says some property is satisfied, then it is intended to mean that $a$ claims \emph{they believe} the property is satisfied. One may claim the latter is a weaker statement, but at the very least the there exists an implication, which we express as a modal axiom.
\begin{itemize}
	\item For any $A$, $\I_{a\ot b} A \Rightarrow \I_{a\ot b}\B_b A$.
\end{itemize}

\subsection{Optional hierarchy of agents} 

We consider one more type of axiom we may add in order to increase the power of the logic.
These are associated to the definition of \emph{shared domains}, 
which combine several domains into one in order to evaluate consequences of their combined knowledge.
We define a preorder $\sqsubseteq$ on agents $a \sqsubseteq b$ denoting that $b$ inherits all knowledge of~$a$:
\begin{itemize}
	\item If $a \sqsubseteq b$ then $\B_a A \Rightarrow \B_b A$ for any $A$.
	\item If $a \sqsubseteq c$ and $b \sqsubseteq d$ then $\I_{a \ot b} A \Rightarrow \I_{c \ot d} A$ for any $A$.
\end{itemize}
Consequently, $\B_a A \Rightarrow \B_b \B_c A$ if $a \sqsubseteq b$ and $a \sqsubseteq c$,
and $\I_{a \ot c} A \Rightarrow \B_b\I_{a \ot c} A$ if $a \sqsubseteq b$.

As an example, we can take our domains to be subsets of entities, $\Ag = \mathcal{P}(\Ag')$, and let $S \sqsubseteq Z$ if $S$ is a subset of $Z$.
For $S \subseteq \Ag'$ we have that $\B_S$ is the combined belief of all entities $a \in S$.
We can for instance show that $\B_{\{a\}}A \Rightarrow \B_{\{b\}}(A \Rightarrow B) \Rightarrow \B_{\{a,b\}}(B)$.
Though we have that $\B_{\{a\}}A \vee \B_{\{b\}}A \Rightarrow \B_{\{a,b\}}(A)$, the converse does not hold as we cannot prove $\B_{\{a\}}A \Rightarrow \B_{\{b\}}(A \Rightarrow B) \Rightarrow (B_{\{a\}}(B) \vee \B_{\{b\}}(B))$.

We can also consider a special domain marking knowledge accessible to all.
This is in line with the traditional necessity modality $\Box$ which marks properties which are certainly true. We can incorporate this in our scheme by considering a special agent $\Omega \in \Ag$ for which $\Omega \sqsubseteq a$ for all other agents and assert the above axiom. This makes $\I_{\Omega \ot a}$ into a kind of \emph{broadcast} modality, as in what $a$ makes public.

One the other side of the spectrum, we could consider the theoretical gathering of all perspectives, which can be used to check if there are any inconsistent opinions. This would be given by some agent $\Upsilon$ such that $b \sqsubseteq \Upsilon$ for any other agent~$b$.
In terms of our $\Ag = \mathcal{P}(\Ag')$ example, $\Omega = \emptyset$ and $\Upsilon = \Ag'$.

Last but not least, we could consider the $\Box$ modality itself as well, to easily mark assumptions which we want all domains to adopt. We assume that $\Box A \Rightarrow \Box \Box A$, $\Box A \Rightarrow A$ and $\Box A \Rightarrow \M A$ for any other modality $\M$.

\begin{figure}
	\begin{center}
\begin{tabular}{| l | l |}
	\hline
	Axiom & Name \\
	\hline
	$\vdash \M A \Rightarrow \M (A \Rightarrow B) \Rightarrow \M B$ & Axiom K \\
	If $\vdash A$ then $\vdash \M A$ & Necessity \\
	\hline
	$\vdash \B_a A \Rightarrow \B_a \B_a A$ & Self awareness \\
	$\vdash \I_{a \ot b} A \Rightarrow \B_a \I_{a \ot b} A$ & Recipient awareness \\
	$\vdash \I_{a \ot b} A \Rightarrow \I_{a \ot b}\B_b A$ & Intent of claim \\
	\hline
	If $a \sqsubseteq b$ then $\vdash \B_a A \Rightarrow \B_bA$ & Belief inheritance \\
	If $a \sqsubseteq b$ then $\vdash \I_{a \ot c} A \Rightarrow \I_{b \ot c} A$ & Receiver inheritance \\
	If $a \sqsubseteq b$ then $\vdash \I_{c \ot a} A \Rightarrow \I_{c \ot b} A$ & Sender inheritance \\
	\hline
	$\vdash \Box A \Rightarrow \M \Box A$ & Public awareness \\
	$\vdash \Box A \Rightarrow A$ & Public verifiability\\
	\hline
\end{tabular}
\end{center}
\caption{Axiom system. $\M$ ranges over modalities, $a$, $b$ over agents, and $A$, $B$ over formulas.}
\label{fig:axioms}
\end{figure}

\subsection{Logical Foundation}

We define the set of formulas $\Form$ inductively as follows:
\[
A, B := \top \mid \bot \mid t \mid \M A \mid A \Rightarrow B \mid A \wedge B \mid A \vee B
\]
Note that by extension, we can formulate conjunction and disjunction over finite sets of formulas as well.

We consider a finite set of modalities $\Mes$ satisfying axioms K and necessity, and other axioms given in the following form:

\begin{definition}
	A modal unfolding axiom is given by a pair $(\M \Rrightarrow l)$ where $\M \in \Mes$ is a modality, and $l \in \Mes^*$ is a list of modalities. 
\end{definition}

For each $l \in \Mes^*$ and formula $A$ we define $l(A)$ inductively as $\varepsilon(A) = A$ and $(\M \cdot l)(A) = \M (l (A))$. The unfolding axiom $\M \Rrightarrow l$ denotes the axiom $\M X \Rightarrow l(X)$.

Given some set of unfolding axioms $\mathsf{Ax}$, we define our logic to be the intuitionistic multimodal logic, where each modality satisfies axiom K and necessity, and where furthermore the axioms from $\mathsf{Ax}$ are asserted. We denote this $\mathcal{L}_{\mathsf{Ax}}$. In the rest of the paper, we simply take the sets of assertions $\mathsf{Ax}$ as discussed before, summarised in Figure \ref{fig:axioms}. 

There are many ways to set up such a logic. In Figure \ref{Fig:basic} is a setup for for a sequent calculus for the logic, with proof steps. A \emph{sequent} is a pair $\Gamma \vdash A$ consisting of a set of formulas $\Gamma$ denoting the assumptions, and a formula $A$ denoting the consequent. The fractions in the figure, called \emph{judgements}, tell us how to construct new true sequents out of known true sequents. The calculus is can be improved to better suit applications, and is mostly included here for reference. Later, in Figure \ref{fig:Fitch}, we formulate a Fitch-style lambda-calculus which forms a fragment of dependent modal type theory \cite{MDTT}, and is more practical for optimizing proofs\footnote{We shall consider the Proofs=Programs Curry-Howard correspondence there. What we are mainly interested in in the formulation of that variant is to get the \emph{substitutivity} property, which allows us to simplify proofs.}.

\begin{figure}
	\[
		\frac{}{\Gamma , A \vdash A}
		\qquad 
		\frac{}{\Gamma \vdash \top}
		\qquad 
		\frac{\Gamma \vdash \bot}{\Gamma \vdash A}
	\qquad
		\frac{\Gamma \vdash A \qquad \Gamma \vdash A \Rightarrow B}{\Gamma \vdash B}
		\qquad 
		\frac{\Gamma , A \vdash B}{\Gamma \vdash A \Rightarrow B}
	\qquad
		\frac{\Gamma \vdash A \wedge B}{\Gamma \vdash A}
		\qquad 
		\frac{\Gamma \vdash A \wedge B}{\Gamma \vdash B}
	\]
	\[
		\frac{\Gamma \vdash A \qquad \Gamma \vdash B}{\Gamma \vdash A \wedge B}
	\qquad
		\frac{\Gamma \vdash A}{\Gamma \vdash A \vee B}
		\qquad 
		\frac{\Gamma \vdash B}{\Gamma \vdash A \vee B}
	\qquad
		\frac{\Gamma \vdash A \vee B \qquad \Gamma , A \vdash C \qquad \Gamma , B \vdash C}{\Gamma \vdash C}	
	\]
	\[
		\frac{\Gamma \vdash \M A_1 \dots \Gamma \vdash \M A_n \qquad A_1 , \dots , A_n \vdash B}{\Gamma \vdash \M B}
	\qquad
		\frac{\Gamma \vdash \M A \qquad \M \Rrightarrow l}{\Gamma \vdash l(A)}
	\]
	\caption{Intuitionistic Multi-Modal K Logic}
	\label{Fig:basic}
\end{figure}

The rest of this subsection is there to consider \emph{decidability} of the logic. That is, do we have an algorithm to determine whether or not a sequent $\Gamma \vdash A$ is provable. This is very useful in practice, as it allows us to check and verify the truth of properties given certain assumptions. This in turn allows us to hold entities accountable for their claims, as they should be aware of consequences of their claims and can be blamed if these uncover a lie.

\subsection{Decidability}
The main property is that the logic of Figure \ref{Fig:basic} given the modal axioms of Figure \ref{fig:axioms} given that $\sqsubseteq$ is a preorder and $\Ag$ is finite, forms a decidable logic. The rest of this subsection explores the nuances of what kind of axiomatic systems $\mathsf{Ax}$ gives rise to a decidable logic, and can be skipped. More details can be found in \cite{TAB}.

To get decidability, $\mathsf{Ax}$ needs to satisfy three additional properties.

\begin{itemize}
	\item $\mathsf{Ax}$ is \emph{reflexive} if $\M \Rrightarrow \M \in \mathsf{Ax}$ for each modality.
	\item $\mathsf{Ax}$ is \emph{transitive} if for any $\M \Rrightarrow l \cc \N \cc l' \in \mathsf{Ax}$ and $\N \Rrightarrow r \in \mathsf{Ax}$, we have $\M \Rrightarrow l \cc r \cc l' \in \mathsf{Ax}$.
	\item $\mathsf{Ax}$ is \emph{decomposable} if for any $\M \Rrightarrow \N \cc \N' \cc l$ there is an $\R \in \Mes$ such that $\M \Rrightarrow \N \cc \R \in \mathsf{Ax}$ and for any  $\M \Rrightarrow \N \cc \T \cc r \in \mathsf{Ax}$, we have $\R \Rrightarrow \T \cc r \in \mathsf{Ax}$.
\end{itemize}

The first two expose the categorical structure of the axioms, whereas the third property asserts that any ``unfolding'' of a modality into three or more modalities can be factored into a composition of axioms of the form $\M \Rrightarrow \N \cc \R$, as well as such decompositions are universal allowing for easier proof search.
We denote a choice of $\R$ as given in the third point by $\M \ominus \N$, meaning that if $\M \Rrightarrow \N \cc \N' \cc l$, then $\M \ominus \N$ exists, $\M \Rrightarrow \N \cc (\M\ominus\N)$, and $\M \ominus \N \Rrightarrow \N' \cc l$.

The reflexive and transitive closure of $\mathsf{Ax}$ is denoted as $\widehat{\mathsf{Ax}}$.

\begin{lemma}
	For any $\mathsf{Ax}$, the logics $\mathcal{L}_{\mathsf{Ax}}$ and $\mathcal{L}_{\widehat{\mathsf{Ax}}}$ are equivalent.
\end{lemma}

\begin{proof}
	We prove this by induction on the reflexive transitive closure.
	If $\M \Rrightarrow l \in \widehat{\mathsf{Ax}}$ then either:
	\begin{itemize}
		\item $\M \Rrightarrow l \in \mathsf{Ax}$ in which case we are done.
		\item $l = \M$, in which case the axiom asserts that $\M A \Rightarrow \M A$ for any formula $A$, which is trivially true.
		\item It came from the composition of $\M \Rrightarrow l_1 \cc \N \cc l_2$ and $\N \Rrightarrow r$, upon which we can apply the induction hypothesis. So for any formula $A$, we have that $\M (A) \Rightarrow l_1 \N l_2 (A)$ and $\N(l_2 A) \Rightarrow r(l_2 A)$ are provable. Applying necessity on the latter, we get $l_1 \N l_2 A \Rightarrow l_1 r l_2 A$, and hence by intuitionistic reasoning ($A \Rightarrow B, B \Rightarrow C \vdash A \Rightarrow C$) we get that $\M A \Rightarrow l_1 r l_2 A$ is provable.
	\end{itemize}
\end{proof}

\begin{lemma}
	If $\textsf{Ax}$ is the system of axioms defined in Figure~\ref{fig:axioms}, then $\widehat{\textsf{Ax}}$ is decomposable.
\end{lemma}

\begin{proof}
	For $\B_a$, note that $\B_a \Rrightarrow \N \cc \N' \cc l \in \widehat{\textsf{Ax}}$ means $\N \cc \N' \cc l = \B_{a_1} \cc \B_{a_2} \cc \dots \cc \B_{a_n}$ for some list of at least two agents such that $a \sqsubseteq a_i$ for all $i\in\{1,\ldots,n\}$ (without inheritance, all $a_i$ are $a$).
	We simply take $\B_a \ominus \B_{a_1} = \B_a$, for which $\B_a \Rrightarrow \B_{a_1} \cc \B_a$ and $\B_a \Rrightarrow \B_{a_2} \cc \dots \cc \B_{a_n}$.
	
	For $\I_{a \ot b}$, axiom $\I_{a \ot b} \Rrightarrow \N \cc \N' \cc l \in \widehat{\textsf{Ax}}$ means $\N \cc \N' \cc l = \B_{a_1} \cc \dots \cc \B_{a_n} \cc \I_{a' \ot b'} \cc \B_{b_1} \cc \dots \cc \B_{b_m}$ where $n + 1 + m \geq 2$, $a \sqsubseteq a_i$, 
	$a \sqsubseteq a'$, $b \sqsubseteq b'$ and $b \sqsubseteq b_j$.
	If $n = 0$ we simply take $\I_{a \ot b} \ominus \I_{a' \ot b'} = \B_b$, and if $n > 0$ we take $\I_{a \ot b} \ominus \B_{a_1} = \I_{a \ot b}$.
\end{proof}

\section{Chains of Communications}\label{sec:chains}

Modalities let us mark different domains of knowledge, and are as such useful for describing distributed networks of trust. In such networks, entities make claims, which may be communicated throughout the network, passing through multiple entities. At each location, it can then be determined what is known and what can be derived.

We consider a canonical example, which we shall describe in a multitude of ways. The scenario is that a certain claim goes through multiple entities before getting to the targeted recipient. How is this claim sent, and what are the preconditions for the recipient to trust this claim?

We consider four entities in a chain of communications:
\[
\xymatrix{
a \ar[r] & b \ar[r] & c \ar[r] & d 
}
\]
Here $a$ will function as the origin of some claim $\p$, and $d$ as the intended recipient.
However, suppose either $a$ cannot send the claim directly to $d$, or $d$ does not trust $a$ directly. Whichever is the case, $b$ and $c$ are used as intermediaries.

\begin{example}[A global view]
	The simplest model simply asserts that the intermediaries pass along the the claim as their own.
	This can be stated by $\I_{b \ot a}\p \Rightarrow \I_{c \ot b}\p$ and $\I_{c \ot b}\p \Rightarrow \I_{d \ot c}\p$. They say that if $b$ receives claim $\p$ from $a$, it will send claim $\p$ to $c$, and similarly for $c$.
	The assertions compose into $\I_{b \ot a}\p \Rightarrow \I_{d \ot c} \p$, hence if $a$ sends the claim it will end up at $d$.
\end{example}

The assertions in the above examples do require some trust by $b$ and $c$, since they need to assert the truth of $\p$. Secondly, note that we need access to all such assertions in order to make the final derivation. We shall address both these things in the next two examples.

\begin{example}[Trust and Responsiveness]\label{ex2}
	The statement $\I_{b \ot a}\p \Rightarrow \I_{c \ot b}\p$ requires $b$ to repeat a claim to another, which effectively means $b$ must state that $\p$ is true. Hence the assertion requires trust of $b$ in $a$. Moreover, the statement asserts a response by $b$, the sending of information to $c$.
	We can break it up in the following components:
	\begin{itemize}
		\item Trust: $\B_b(\I_{b \ot a} \p \Rightarrow \p)$, if $b$ receives the claim $\p$ from $a$, then it is believed. Note that with axiom $\I_{b \ot a} A \Rightarrow \B_b \I_{b \ot a} A$, we may derive $\I_{b \ot a} \p \Rightarrow \B_a \p$.
		\item Responsiveness: $\B_b(\I_{b \ot a} \p \wedge \p) \Rightarrow \I_{c \ot b} \p$, if $b$ receives the claim $\p$ and believes it, $b$ will send it to $c$. We can split this up further into two statements.
		\begin{itemize}
			\item $\B_b(\I_{b \ot a} \p \wedge \p \Rightarrow \I_{c \ot b} \p)$ where $b$ believes in their own commitment to sending the claim along, and
			\item $\B_b\I_{c \ot b} \p \Rightarrow \I_{c \ot b} \p$ saying we trust $b$ keeps their commitment\footnote{One could argue this should be an axiom, though note it requires $b$ to actively check and act on their commitments, hence it is reasonable to formulate it as a kind of trust in $b$ instead of a separate axiom.}.
		\end{itemize}
	\end{itemize}
\end{example}

Suppose we want to ensure that $d$ can derive that the message is passed along. The assertions defined prior as stated are not accessible to $d$, by which we mean we have no reason for thinking $d$ knows or believes in them.
So if $d$ does not have them as prior assumptions, the entities are required communicate these statements of trust and responsiveness to $d$ themselves.
For brevity, we write $\I_{a_1 \ot a_2 \ot \dots \ot a_n}$ for $\I_{a_1 \ot a_2}\I_{a_2 \ot a_3}\dots \I_{a_{n-1} \ot a_n}$.

\begin{example}[Communicating Trust]
	Instead of stating $\I_{b \ot a}\p \Rightarrow \I_{c \ot b}\p$, we let $b$ tell $c$ that they will do this. We state that as $\I_{c \ot b}(\I_{b \ot a}\p \Rightarrow \p)$, meaning $b$ tells $c$ that if they receive claim $\p$ from $a$, then $\p$ is true. Agent $b$ need not explicitly send claim $\p$ anymore, they can simply forward a received claim from $a$ and thereby implying the truth of $\p$. This consequence is derived by applying axiom K, getting $\I_{c \ot b}\I_{b \ot a}\p \Rightarrow \I_{c \ot b} \p$.
	
	With everything being forwarded to $d$, we get statements:
	\begin{itemize}
		\item $\I_{d \ot c \ot b}(\I_{b \ot a} \p \Rightarrow \p)$, it is forwarded that $b$ trusts $a$.
		\item $\I_{d \ot c}(\I_{c \ot b} \p \Rightarrow \p)$, it is communicated that $c$ trusts $b$.
		\item $\B_d(\I_{d \ot c} \p \Rightarrow \p)$, $d$ trusts $c$.
	\end{itemize}
	All composed together, we can derive $\B_d(\I_{d \ot c \ot b \ot a} \p \Rightarrow \p)$, $d$ trusts forwarded claims from $a$.
\end{example}

The communicated statements of the previous example can be split up into different parts as done in example \ref{ex2}, if more clarification is needed. There is still one disadvantage regarding the above approach: each link in the chain must trust the previous link regarding statement $\p$, even though that previous link may not be an expert on $\p$. This leaves the system open to a kind of injected claim attack: if any link starts claiming $\p$, everyone down the line will start to accept this claim, even though the claim did not originate from $a$.
Though this problem cannot be entirely prevented, as we need to assume at least some integrity from the links, there are three partial solutions we shall consider.

The first solution is to fine tune which claims are trusted. Instead of simply trusting $A$, we can only accept certain evidence of $A$. Even though a link may not be an expert on $A$, it could at least be trusted regarding other things, like communicated claims from other experts and verifying that such claims came from trusted individuals. This shall be the topic of Section~\ref{sec:trust}.
The second solution is to use multiple sources to double check claims, which we shall consider in Section \ref{sec:threshold} using \emph{threshold trust}.

The third solution has a different flavour, with none of the links in the communication chain needing to trust each other. Instead, they all have access to a single trusted source which can be consulted and used to double check if a claim is correct.
This source can for instance be some kind of internet server, or a certification authority.

\begin{example}[Trusted Authority]
	Let $v$ be a new entity everyone has access to.
	
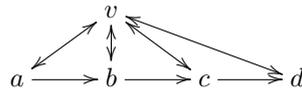
\[
		\xymatrix@R=5mm{
			& v \ar@{<->}[dl] \ar@{<->}[d] \ar@{<->}[dr] \ar@{<->}[drr] & & \\
			a \ar[r] & b \ar[r] & c \ar[r] & d 
		}
	\]
	The behaviour of $v$ can be specified as follows: it can send out claims (e.g. certifications) and will confirm to others that they indeed made this claim.

	This is given as follows: $\bigwedge_{x,y \in \Ag} \I_{v \ot x \ot y \ot v} \p \Rightarrow \B_v \p \Rightarrow \I_{x \ot v}\p$, meaning if someone notifies $v$ that they heard from someone else that $v$ has claimed something, then if this is indeed true, then $v$ confirms this claim.
	
	This can then be used by $b$ given two statements:
	\begin{enumerate}
		\item $\I_{b \ot a \ot v} \p \Rightarrow \I_{v \ot b \ot a \ot v} \p \vee \I_{b \ot v} \p$, which says that if $b$ receives from $a$ a claim made by $v$, it will check this with $v$ directly, unless it has already received this claim from $v$. Note that with the specification of $v$, this implies $\I_{b \ot a \ot v} \p \Rightarrow \B_v \p \Rightarrow \I_{b \ot v} \p$.
		\item $\I_{b \ot v} \p \Rightarrow \I_{c \ot b \ot v} \p$.
	\end{enumerate}
	Supposing $b$, $c$ and $d$ have the first property, and $a$, $b$ and $c$ have the second property, we can derive $\B_v \p \Rightarrow \I_{a \ot v} \p \Rightarrow \I_{d \ot v} \p$.
\end{example}

\section{Trust}\label{sec:trust}

For cooperation, it is useful for entities to declare their trust in someone before that entity has made the relevant claim. This way, others can be assured that the entity can be convinced if needed, and how this can be done.
We define trust in two stages. First we introduce some notation; for each list of modalities $l = \M_1,\dots,\M_n$ and formula $A$, \emph{validity} of $l$ concerning $\p$ is given by:
\[
	\widehat{(\M_1 \dots \M_n)}\p := \M_1 \dots \M_n \p \Rightarrow \p
\]
For instance, validity $\widehat{\I_\alpha}$ for some $\alpha = a_1 \ot a_2 \ot \dots \ot a_n$ is the assertion that if claim $\p$ is communicated through chain of communications $\alpha$, then $\p$ is true. Hence validity depends on four aspects: the sender, the receiver, the route the message took, and the content of the claim. 

The concept of \emph{trust} is simply a statement of validity put into some context.
Validity itself is the expression that ``we", meaning those considering the current set of assumptions, trust the claim made over some channel. We can also consider trust as either believed or claimed by a particular entity. E.g. $\B_b \valu{\alpha} \p$ means $b$ trusts claim $\p$ if communicated over $\alpha$, and $\I_{\beta} \valu{\alpha}$ is this trust being proclaimed over the chain of communications $\beta$. Usually, trust occurs in the recipient of the considered claim, e.g. in $\B_b \valu{b\ot\alpha} \p $ or $\I_{\gamma \ot b} \valu{b\ot \alpha}A$.

Given our axioms, \emph{cautious trust} $\T_{a,b} \p := \B_a(\I_{a \ot b}\p \Rightarrow \B_b \p \wedge \B_b \p \Rightarrow \p)$ as formulated by Liau \cite{LIAU2003} can be proven with the following validities:
\[
\B_a \valu{a \ot b} \B_b \p \wedge \B_a \valm{\B_b} \p  \Rightarrow \T_{a , b} \p
\]

$\valu\alpha$ satisfies two rules which helps reasoning about it:
\begin{itemize}
	\item $\p \Rightarrow \valu\alpha \p$, if a statement $\p$ is already true, then communications making this claim are trivially valid.
	\item $\valu\alpha A \wedge \valu{\alpha} B \Rightarrow \valu\alpha (A \wedge B)$, a combined statement's validity can be determined in parts. The converse does however not hold.
\end{itemize}

\begin{lemma}\label{lem:counter}
	For any modality $\M$, and formulas $A$ and $B$, $\valm{\M}A \wedge \valm{\M}(A \Rightarrow B) \Rightarrow \valm{\M}B$ is not necessarily provable.
\end{lemma}
\begin{proof}
	We show that there is a formula $A$ and a Kripke frame which does not satisfy this statement.
	Let $W$ have two worlds, $v$ and $w$, and let $\leq$ be the identity relation.
	Let $t$ and $r$ be two tokens such that $S_t = \emptyset$ and $S_r = \{w\}$.
	Lastly, let $R_\M$ relate $v$ with $w$, but not $v$ with $v$.
	
	Then neither $v$ nor $w$ model $\M t$, hence both model $\widehat{\M}t$.
	Neither $v$ nor $w$ model $t$, so both model $t \Rightarrow r$, and hence $\widehat{\M}(t \Rightarrow r)$.
	Since $v R_\M v$ does not hold, and $w \models r$, we have $v \models \M r$. However, since $v$ does not model $r$, $v$ does not model $\widehat{M}r$.
	So $v$ does not model $\valm{\M}t \wedge \valm{\M}(t \Rightarrow r) \Rightarrow \valm{\M}r$.
\end{proof}

\begin{lemma}\label{lem:comp}
\[
\valu{\alpha \ot b} \p \wedge \I_{\alpha \ot b} \valu{b \ot \gamma} \p \Rightarrow \valu{\alpha \ot b \ot \gamma} \p
\]
\end{lemma}
\begin{proof}
	$\I_{\alpha \ot b} \valu{b \ot \gamma} \p = \I_{\alpha \ot b}(\I_{b \ot \gamma} \p \Rightarrow \p)$ implies $\I_{\alpha \ot b}\I_{b\ot \gamma} \p \Rightarrow \I_{\alpha \ot b}\p$ by axiom K, where $\I_{\alpha \ot b}\I_{b \ot \gamma} = \I_{\alpha \ot b \ot \gamma}$.
	Hence combined with $\valu{\alpha \ot b} \p = (\I_{\alpha \ot b} \p \Rightarrow \p)$ we get $(\I_{\alpha \ot b \ot \gamma} \p \Rightarrow \p) = \valu{\alpha \ot b \ot \gamma} \p $.
\end{proof}
\noindent 
Here validity is determined in two stages:
\begin{itemize}
\item $b$ communicates over $\alpha$ that they would trust claim $\p$ if coming over channel $\gamma$.
\item Communication of $\p$ over $\alpha$ are trusted.
\item As such, communications over $\gamma$ through $b$ and then over $\alpha$ are trusted.
\end{itemize}

We gather derivable properties of validities studied in this section in Figure \ref{fig:validityproperties} for reference.

\begin{figure}
	\begin{center}
		\begin{tabular}{| l | l | }
			\hline
			Property & Name \\
			\hline
			$\vdash \widehat{\M}A \wedge \M A \Rightarrow A$ & Validity definition \\
			$\vdash A \Rightarrow \widehat{\M}A$ & Validity unit \\
			$\vdash \widehat{\M}A \wedge \widehat{\M}B \Rightarrow \widehat{\M}(A \wedge B)$ & Merging validities \\
			$\vdash \widehat{\M}\N A \wedge \widehat{\N} A \Rightarrow \widehat{\M\N} A$ & Composition rule 1\\
			$\vdash \M\widehat{\N} A \wedge \widehat{\M} A \Rightarrow \widehat{\M\N} A$ & Composition rule 2\\
			$\vdash \M\widehat{\N} A {\wedge} \widehat{\M}\widehat{\N} A {\wedge} \widehat{\M}\N A {\Rightarrow} \widehat{\M\N} A$ & Composition rule 3\\
			$\not\vdash \widehat{\M}A \wedge \widehat{\M}(A \Rightarrow B) \Rightarrow \widehat{\M}B$ & No axiom K \\
			\hline
		\end{tabular}
	\end{center}
	\caption{Validity properties (Derivable in the logic)}\label{fig:validityproperties}
\end{figure}

\subsection{Indirect trust}

Let us look at a basic chaining of trust.
We shall reason from the perspective of an agent $a$, though we could as easily have considered $a$ communicating these derivations to other agents instead.
The basic derivation done by applying axiom K to Lemma \ref{lem:comp} is as follows:
\[
\B_a \valu{a \ot b} \p  \wedge  \I_{a \ot b} \valu{b \ot c} \p \Rightarrow \B_a \valu{a \ot b \ot c} \p
\]
Here $a$ reasons that they trust $b$ regarding property $\p$. If $b$ then communicates they trust $c$ regarding property $\p$, $a$ can infer trust in $c$ if at least claims go through $b$. This is an effective way of chaining trust, which requires $a$ to trust $b$ regarding $\p$ without argument. Note however that we also have the following derivation:
\[
\B_a\valu{a \ot b} \p \wedge \I_{a \ot b} \p \implies \B_a \p
\]
In certain situations, we may want to fine-tune and reduce the necessary trust required of $a$.

Let us instead suppose $a$ accepts a certain proof from $b$, but does not want to simply accept the conclusion. The argument from $b$ comes in two parts: 
\begin{enumerate}
	\item I (meaning $b$) trust $c$ if they think $\p$ is true.
	\item $c$ has communicated to me that $\p$ is true.
\end{enumerate}
The obvious conclusion of these two statements is that $b$ believes $\p$ is true. Instead of simply accepting this conclusion, $a$ can choose to only accept the arguments themselves. They can state the following:
\[
\B_a \valu{a \ot b} \valu{b \ot c} \p  \wedge \B_a \valu{a \ot b} \I_{b \ot c} \p 
\]
stating that they trust $b$ regarding their trust in $c$, as well as trust $b$ to not falsely forward messages from $c$.
It should be noted that $a$ can derive $\B_a \valu{a \ot b} (\valu{b \ot c} \p  \wedge \I_{b \ot c} \p)$ but cannot derive $\B_a \valu{a \ot b} \p$ as seen in Lemma \ref{lem:counter}. The arguments are accepted, but not the conclusion.
Still, $b$'s declaration of trust allows $a$ to gain trust as well:
\[
\B_a \valu{a \ot b} (\valu{b \ot c} \p  \wedge \I_{b \ot c} \p) \wedge \I_{a \ot b}\valu{b \ot c} \Rightarrow \B_a\valu{a \ot b \ot c} \p
\]
The argument goes as follows: $b$ has already communicated its trust to $a$. So if $b$ then later mentions that $c$ claimed $\p$, the two claims by $b$ can be combined into a single package $\I_{a \ot b}(\valu{b \ot c} \p \wedge \I_{b \ot c}A)$ which is trusted by $a$, and hence $\B_a (\valu{b \ot c} \p \wedge \I_{b \ot c}A)$. Having accepted the arguments then allows $a$ to make the appropriate derivation, that $\B_a A$.
However, if in this case $b$ had communicated $\p$ directly, $a$ would not have accepted it: $a$ trusts $b$ regarding identifying experts on $\p$, but $a$ does not trust $b$ as an expert on $\p$ itself.

In general, suppose $b$ has derived a statement $B$ using assumptions $A_1 , \dots , A_n$. People like $a$ can trust $b$'s conclusion directly, as stated by $\B_a \valu{a \ot b} B$. Or $a$ can instead require that $b$ provide their arguments, and state $\B_a \valu{a \ot b}(A_1 \wedge \dots \wedge A_n)$ or $\B_a (\valu{a \ot b}A_1 \wedge \dots \wedge  \valu{a \ot b}A_n)$. In the latter situation, $b$ can only convince $a$ by claiming all the arguments. There are still two things to be considered.

First note that the proof of $B$ via $A_1 , \dots , A_n$ is never communicated. The proof is implied, and has to be done by $a$. Indeed, the burden of proof is naturally on the one who wants to trust, not on the entity providing the facts. One could imagine $a$ asking $b$ whether $B$ is true, and $b$ communicating $A_1 \wedge \dots \wedge A_n$ instead of $B$ itself. Regardless of whether $a$ trusts $b$ regarding $B$, or regarding $A_1 \wedge \dots \wedge A_n$, $a$ will be able to convince themselves of $B$.
The fact that $B$ is the conclusion is implied by the claims.
Note that with decidability, $a$ will be able to verify that the proof exists.

Secondly, it may not make sense for $a$ to trust $b$ regarding just a specific entity $c$. In most instances, the particular $c$ in question may not be predetermined. We can state that $a$ trusts $b$ regarding any such entity, as can be formulated as:
\[
\B_a \bigwedge_{c \in \Ag} \valu{a \ot b} (\valu{b \ot c} \p \wedge \I_{b \ot c} \p)
\] 
We call this \emph{higher-order trust}. We do not trust $b$ directly regarding $\p$, but do trust $b$ to identify experts of $\p$.
These orders of trust are an important tool for describing trust infrastructures.

\section{Forwarding Networks}\label{sec:forward}

We use our logic to describe trust in networks of entities. These can be used to describe public key infrastructures, as well as other multiparty networks of authentication.
The core idea is to establish what the role of an entity in the network is, and what type of trust is required for productive cooperation with this entity. 
Agents may establish trust directly, or use intermediaries like certification authorities to establish trust.

We start by describing the language on a high level.
We write $\mtrust{n}{a}{b}{\tP}$ to say that $a$ trusts $b$ regarding predicate $\tP : \Ag \to \Form$, to a certain \emph{order} $n$. Here, $\tP$ is some predicate on agents, describing some fundamental claim one wants to convey to others, like whether an agent owns a certain public key, or whether an agent has a certain certificate or diploma\footnote{$\tP$ need not depend on $a$, and can be some other primitive statement.}.
Most likely, we take as our set of tokens $\Tok$ some statements that can refer to agents, e.g. with $\texttt{hasKey}(a) \in \Tok$ for each $a \in \Ag$, we can use $\texttt{hasKey}$ as a predicate.
The order $n$ is a natural number saying whether there is direct trust between two agents, or wether there is trust between agents as intermediaries in bigger networks.
\begin{itemize}
	\item $\mtrust{0}{a}{b}{\tP}$ means that if $a$ receives a claim $\tP(b)$ from $b$, then $a$ will accept this claim $\tP(b)$ as true.
	\item $\mtrust{n+1}{a}{b}{\tP}$ means that if $a$ receives a claim from $b$ regarding the $n$-th order trust of $b$ in $\tP$, then $a$ believes this trust is warranted.
\end{itemize}

We can look at some examples of different orders of trust, in the realm of public key infrastructures.
\begin{itemize}
	\item Order 0 trust: This is the goal. Certification authorities record and communicate their order 0 trust regarding key ownership, and others can inherit this order 0 trust.
	\item Order 1 trust: This is the order of trust one has in a certification authority. With order 1 trust in the CA, one can inherit their order 0 trust.
	\item Order 2 trust: This is the trust one has in authorities which validate CAs, like domain trusted lists. With order 2 trust in a domain trusted list, one can inherit order 1 trust in the CA which is listed.
	\item Order 3 trust: This is trust in entities which claim to have order 2 trust. For instance, trust in a government backing some domain trusted list. 
\end{itemize}


In order to establish trust across networks, claims will need to be forwarded and checked by intermediary agents.
We have to specify the allowed paths of communication.
\begin{definition}
	A \emph{forwarding network} is a set of non-empty, non-repeating lists of agents $S \subseteq \Ag^*$, such that:
\begin{enumerate}
	\item If $(a , \beta , c) \in S$ then $(\beta , c) \in S$ and $(a,\beta) \in S$ (in other words, $S$ is closed under taking sublists).
	\item If $(\alpha , a , \beta , b , \gamma) \in S$ and $(a , \beta' , b) \in S$ then $(\alpha , a , \beta' , b , \gamma) \in S$.
\end{enumerate}
\end{definition}
If $a , \alpha , b \in S$, we say that $\alpha$ is a defined path from $b$ to $a$. Any segment of the path can be replaced by an alternative route, if specified by $S$. Note that $\alpha$ can be empty, in which case there is a direct path.

We write $\N_S(a,b,c)$ if there is a path $a, \alpha , b , \beta , c \in S$.
\begin{lemma}\label{lem:Nfac}
	Given a forwarding network $S \subseteq \Ag^*$, then $\N_S$ has the following properties:
\begin{itemize}
	\item If $\N_S(a,b,d) \& \N_S(b,c,d)$ then $\N_S(a,b,c) \& \N_S(a,c,d)$.
	\item If $\N_S(a,b,c) \& \N_S(a,c,d)$ then $\N_S(a,b,d) \& \N_S(b,c,d)$.
\end{itemize}
\end{lemma}

\subsection{Forwarding Modality}

It is assumed that messages will be passed along across every specified path. This is independent of whether the message itself is trusted. As we have seen before, we can establish commitments for forwarding messages which do not require additional trust from the network members, except for trusting they will pass along the message. If you forward a message, you are not necessarily claiming that the content is true.
We define the collection of all forwardings with a \emph{composite modality}. For $a , b \in \Ag$, let \[
	\J_{a \ot b}A := \bigwedge\{\I_{a \ot \gamma \ot b}A \mid \gamma \in \Ag^* \text{ such that } (a,\gamma,b) \in S\}
\]
Note that $\J_{a \ot a}A = \top$ by definition.
The composite modality behaves like a modality itself:
\begin{lemma}\label{lem:Jaxiom}
The following properties hold:
\begin{itemize}
	\item $\J_{a \ot b}$ satisfies axiom K and necessity.
	\item $\J_{a \ot b}A \implies \B_a\J_{a \ot b}A$.
	\item $\J_{a \ot b}A \implies \J_{a \ot b}\B_bA$.
\end{itemize}
\end{lemma}

Last but certainly not least, we have the additional property which allows us to factor these modalities further.

\begin{lemma}
	If $\N_S(a, b, c)$ then $\vdash \J_{a \ot c}A \implies \J_{a \ot b}\J_{b \ot c}A$.
\end{lemma}
\begin{proof}
	Assume $\J_{a \ot c}A$, hence for any $a , \alpha , c \in S$ we have $\I_{a \ot \alpha \ot c}A$.
	To prove $\J_{a \ot b}\J_{b \ot c}A$ we prove $\I_{a \ot \beta \ot b}\J_{b \ot c}A$ for each $(a,\beta,b) \in S$.
	$\I_{a \ot \beta \ot b}\J_{b \ot c}A$ in turn can be shown by proving $\I_{a \ot \beta \ot b}\I_{b \ot \gamma \ot c}A$ for each $(b, \gamma, c) \in S$ and applying axiom K, using that there are only finitely many possible $\gamma$-s.
	
	So suppose $(a,\beta,b) \in S$ and $(b,\gamma,c) \in S$, and suppose $N_S(a,b,c)$, meaning $(a,\beta',b,\gamma',c) \in S$ for some $\beta'$ and $\gamma'$.
	By applying property 2 of forwarding networks twice, $(a,\beta,b,\gamma,c) \in S$. 
	Since $\J_{a \ot c}A$ implies any path from $c$ to $a$, it implies $\I_{a \ot \beta \ot b \ot \gamma \ot c}A$ which by definition is $\I_{a \ot \beta \ot b}\I_{b \ot \gamma \ot c}A$, what we need.
\end{proof}

Basically, since $\J_{a \ot c}$ attempts all paths of communication from $c$ to $a$, these include all paths through agent $b$. 
We again use $\valw{a \ot b}$ to express validity of such communications:
\[
	\valw{a \ot b}A := \J_{a \ot b}A \Rightarrow A
\]
Note that this considers all possible paths, so such validity can be obtained if validity exists across any one path.
\begin{itemize}
	\item If $(a, \alpha , b) \in S$, then $\valu{a \ot \alpha \ot b}A \implies \valw{a \ot b}A$.
\end{itemize}

Lastly note that if no path from $b$ to $a$ exists, then $\valw{a \ot b}A \equiv (\top \Rightarrow A) \equiv A$, hence trust across non existing paths only happens when the claim in question is true.

\subsection{Orders of trust}

Direct trust from $a$ in $b$ can be defined as $\B_a\widehat{\I_{a \ot b}}\tF(b)$, meaning $a$ trusts $b$ regarding claim $\tF(b)$.
Agent $b$ is effectively claiming that predicate $\tF$ holds for them, for instance saying: ``I own public key $k$'', or ``I trust $c$''.

We consider the following two properties fundamental to the development. If $\N_S(a,b,c)$ then:
\begin{itemize}
	\item $\J_{a \ot b}\valw{b \ot c}A \wedge \valw{a \ot b} A \implies \valw{a \ot c}A$
	\item $\J_{a \ot b}\valw{b \ot c}A \wedge \valw{a \ot b}\valw{b \ot c}A \wedge \valw{a \ot b}\J_{b \ot c}A \implies \valw{a \ot c}A$
\end{itemize}
Note that the proof of the latter statement first deduces $\valw{a \ot b}A$ from the pieces of evidence $\valw{a \ot b}\valw{b \ot c}A$ and $\valw{a \ot b}\J_{b \ot c}A$, and then applies the former statement.


In the above statements, $A$  can either be a basic statement $\tF(c)$, or itself a statement necessary for composing trust further.
As such, we need sequences of communication and validity claims, whose length depends on the order of trust we want to express. Using $\tF$ as a basis, we define a set of higher-order statements $\textbf{C}^n_a(\tF)$ inductively as follows:~
\[
\textbf{C}^0_a(\tF) := \{\tF(a)\}
\]
\[
\textbf{C}^{n+1}_a(\tF) := \{\valw{a \ot b}A , \m{J}_{a \ot b}A \mid b \in \Ag, A \in \textbf{C}^{n}_b(\tF)\}
\]

Using the fact that each set is finite, we define higher order trust as a formula using conjunctions.

\begin{definition}\label{def:seman}
	We define $n$-th order validity of $a$ in $b$ as: 
	\[
	\nval{n}{a}{b}{\tF} := \bigwedge \{\valw{a \ot b}A \mid A \in \textbf{C}^n_b(\tF)\}
	\]
	We define $n$-th order trust of $a$ in $b$ as: 
	\[
	\ntrust{n}{a}{b}{\tF} := \B_a\nval{n}{a}{b}{\tF}
	\]
\end{definition}

\noindent
For example, $\ntrust{0}{a}{b}{\tF} \equiv  \B_a\valw{a \ot b}\tF(b)$ and  $\ntrust{1}{a}{b}{\tF} \equiv  \bigwedge_c(\B_a\valw{a \ot b}\valw{b \ot c}\tF(c) {\wedge} \B_a\valw{a \ot b}\m{J}_{b \ot c}\tF(c))$.

Let us establish some needed lemmas provable in the logic.

\begin{lemma}\label{lem:exval}
	If $\N_S(a,b,c)$ then $\vdash \nval{n+1}{a}{b}{\tF} \Rightarrow \valw{a \ot b}(\nval{n}{b}{c}{\tF})$.
\end{lemma}
\begin{proof}
	Note that $\widehat{\J_{a \ot b}}A \wedge \widehat{\J_{a \ot b}}B \Rightarrow \widehat{\J_{a \ot b}}(A \wedge B)$ since $\J_{a \ot b}(A \wedge B) \Rightarrow \J_{a \ot b}A \wedge \J_{a \ot b}B$. These implications can be generalised to conjunctions over finite sets.
	As a consequence, 
	
	\noindent
	$\bigwedge_{A \in \textbf{C}^n_c(\tF)}\valw{a \ot b}\valw{b \ot c}A$ 
	$=$ $\bigwedge_{A \in \textbf{C}^n_c(\tF)}(\m{J}_{a \ot b}\valw{b \ot c}A \Rightarrow \valw{b \ot c}A)$
	
	\noindent
	$\implies$ $(\m{J}_{a \ot b}\nval{n}{b}{c}{\tF} {\Rightarrow} \nval{n}{b}{c}{\tF})$

	$=$ $\valw{a \ot b}(\nval{n}{b}{c}{\tF})$
\end{proof}

We can compose validity in the following way.

\begin{lemma}\label{lem:comp-trust}
	If $\N_S(a,b,c)$, the following properties hold:
	\begin{enumerate} 
		\item $\!\!\vdash \nval{n+1}{a}{b}{\tF} \wedge \nval{n}{b}{c}{\tF} \Rightarrow \nval{n}{a}{c}{\tF}$.
		\item $\!\!\vdash \nval{n{+}1}{a}{b}{\tF} \wedge \m{J}_{a \ot b}\nval{n}{b}{c}{\tF} \Rightarrow \nval{n}{a}{c}{\tF}$.
		\item $\!\!\vdash \nval{n}{a}{b}{\tF} \wedge \m{J}_{a \ot b}\nval{n}{b}{c}{\tF} \Rightarrow \nval{n}{a}{c}{\tF}$.
		\item $\!\!\vdash \ntrust{n{+}1}{a}{b}{\tF} {\wedge} \J_{a \ot b}\nval{n}{b}{c}{\tF} {\Rightarrow} \ntrust{n}{a}{c}{\tF}$.
		\item $\!\!\vdash \ntrust{n}{a}{b}{\tF} \wedge \J_{a \ot b}\nval{n}{b}{c}{\tF} \Rightarrow \ntrust{n}{a}{c}{\tF}$.
	\end{enumerate}
\end{lemma}
\begin{proof}
	Let $A \in \textbf{C}^n_c(\tF)$, we want to show that $\valw{a \ot c}A$ in the first three cases, which means $\J_{a \ot c}A \Rightarrow A$. Since $\J_{a \ot c}A \Rightarrow \J_{a \ot b}\J_{b \ot c}A$ it is sufficient to prove $A$ from $\J_{a \ot b}\J_{b \ot c}A$.
	So assume $\J_{a \ot b}\J_{b \ot c}A$,
	\begin{enumerate}
		\item $\nval{n+1}{a}{b}{\tF}$ implies $\valw{a \ot b}\m{J}_{b \ot c}A$, hence $\J_{b \ot c}A$. From $\nval{n}{b}{c}{\tF}$ we have $\valw{b \ot c}A$, hence $A$.
		\item Like before we can derive $\J_{b \ot c}A$.
		$\nval{n+1}{a}{b}{\tF}$ also implies $\valw{a \ot b}\valw{b \ot c}A$, and from $\m{J}_{a \ot b}(\nval{n}{b}{c}{\tF})$ we have $\m{J}_{a \ot b}\valw{b \ot c}A$, hence $\valw{b \ot c}A$, which together with $\J_{b \ot c}A$ makes $A$.
		\item From $\m{J}_{a \ot b}(\nval{n}{b}{c}{\tF})$ we have $\J_{a \ot b}\valw{b \ot c}A$, hence with $\J_{a \ot b}\J_{b \ot c}A$ we get $\J_{a \ot b}A$.
		From $\nval{n}{a}{b}{\tF}$ we get $\valw{a \ot b}A$, and hence $A$.
		\item Apply axiom K and $\J_{a \ot b}A \Rightarrow \B_a \J_{a \ot b}A$ to property 2.
		\item Apply axiom K and $\J_{a \ot b}A \Rightarrow \B_a \J_{a \ot b}A$ to property 3.
	\end{enumerate}
\end{proof}


Suppose we have some forwarding network $S$,
\begin{definition}
	We say that a trust statement $\ntrust{n}{a}{b}{\tF}$, is \emph{shared over} $S$ if $\ntrust{n}{a}{b}{\tF}$ and $\J_{c \ot b}\nval{n}{a}{b}{\tF}$ hold for any $c \in \Ag$ such that $\N_S(c,a,b)$.
\end{definition}
Sharing trust over a network involves sending your claim of trust to all relevant parties. The forwarding network $S$ specifies everyone who can use your claims and may depend on them, and as such gives a specification of where to send your claims. 

\begin{theorem}\label{lem:shared}
	Given a forwarding network $S$ and $\N_S(a,b,c)$,
	\begin{itemize}
		\item If both $\ntrust{n}{a}{b}{\tF}$ and $\ntrust{n}{b}{c}{\tF}$ are shared over $S$, then $\ntrust{n}{a}{c}{\tF}$ is shared over $S$.
		\item If both $\ntrust{n}{a}{b}{\tF}$ and $\ntrust{n+1}{b}{c}{\tF}$ are shared over $S$, then $\ntrust{n}{a}{c}{\tF}$ is shared over $S$.
	\end{itemize}
\end{theorem}
\begin{proof}
	Suppose $\N_S(a,b,c)$, and both $\ntrust{n}{a}{b}{\tF}$ and $\ntrust{n}{b}{c}{\tF}$ are shared over $S$. Then it holds that $\ntrust{n}{a}{b}{\tF}$ and $\J_{a \ot b}\nval{n}{b}{c}{\tF}$, so by property 5 of Lemma \ref{lem:comp-trust}, $\ntrust{n}{a}{c}{\tF}$.
	Supposing moreover $\N_S(d,a,c)$, then by Lemma \ref{lem:Nfac}, $\N_S(d,b,c)$ and $\N_S(d,a,b)$, so $\J_{d \ot a}\nval{n}{a}{b}{\tF}$ and $\J_{d \ot b}\nval{n}{b}{c}{\tF}$, the latter implying $\J_{d \ot a}\J_{a \ot b}\nval{n}{b}{c}{\tF}$. With axiom K extended to $\J$ by Lemma \ref{lem:Jaxiom}, and by property 3 of Lemma \ref{lem:comp-trust}, we get $\J_{d \ot a}\nval{n}{a}{c}{\tF}$. We conclude that $\ntrust{n}{a}{c}{\tF}$ is shared over $S$.
	The second property proven similarly, using Properties 4 and 2 of Lemma \ref{lem:comp-trust} instead.
\end{proof}

\section{Examples from Public Key Infrastructures}\label{sub:exam}
Suppose we have a finite directed graph $\G$ expressing agents and communication channels between agents. We use this to formulate a forwarding network $S$ in two different ways:
\begin{itemize}
	\item We define $S_\G$ as the set of all shortest paths between vertices, with $(a , \beta , c) \in S_\G$ a shortest path from $c$ to $a$, and $a \in S$ the shortest path from $a$ to $a$.
	\item Alternatively, if $\G$ is acyclic, we can define $S'_\G$ to be the set of all paths which do not repeat vertices.
\end{itemize}

We look at examples of Public Key infrastructures. We use $a \overset{n}{\rightarrow} b$ to say that $b \to a$ is an edge in $\G$ and $\ntrust{n}{a}{b}{\tF}$ is shared over $S_\G$. We use $a \overset{n}{\dashrightarrow} b$ to say we can derive that $\ntrust{n}{a}{b}{\tF}$ is shared over $S_\G$ (though $b \to a$ is not an edge).
We derive $a \overset{n}{\dashrightarrow} b$ statements using Theorem \ref{lem:shared}.

The following variety of examples of public key infrastructures are taken from~\cite{cef2018} and~\cite{intdomPKI}.
In all the examples given below, there is at most one non-repeating path between any two entities which does not revisit a vertex, which is therefore automatically the shortest path.

\paragraph{Dedicated domain PKI} The simplest model of indirect trust is the dedicated domain PKI, where a specific certification authority is used to validate keys. 

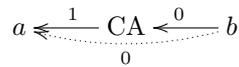
\begin{center}
$
\xymatrix{
	a \ar@{<-}^1[r] \ar@{<.}@/_0.5pc/_0[rr] & \text{CA}  \ar@{<-}^0[r] & b
}
$
\end{center}

A variation of this is the shared domain PKI, where one CA validates keys from different users.
Other systems may use more intermediate entities on the way towards validation. This can be done in two distinct ways, which are shown in the next two examples.

\paragraph{Bridge Certification Authority} In this system, everyone has their own CA which they trust and use to validate their keys. This CA is then linked to a bridge certification authority (BCA), and mutual trust exists between the CAs. When claims are made, each CA in the chain can in turn validate the claim given their trust in each other.
Below is an example of a BCA system, with on the left the assumptions made, and on the right examples of what trust can be derived.

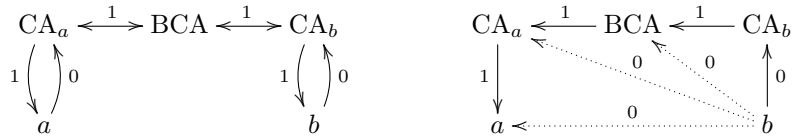
\[
\xymatrix{
	\text{CA}_a \ar@{<->}^1[r] \ar@{<-}@/^0.5pc/^0[d] & \text{BCA} \ar@{<->}^1[r]   & \text{CA}_b \ar@{<-}@/^0.5pc/^0[d]  \\
	a \ar@/^0.5pc/@{<-}^1[u] &  & b \ar@/^0.5pc/@{<-}^1[u]
}
\qquad \qquad
\xymatrix{
	\text{CA}_a \ar@{<-}^1[r]\ar@{<.}^0[drr] & \text{BCA} \ar@{<-}^1[r]\ar@{<.}^0[dr] & \text{CA}_b \ar@{<-}^0[d]  \\
	a \ar@{<-}^{1}[u]  \ar@{<.}^0[rr] & & b
}
\]

\paragraph{Domain Trusted list} Another way of composing trust is by using a domain trusted list. This is different from a CA, as it validates CAs themselves, not ownership of keys directly. Hence, trust in a domain trusted list is of a higher order then trust in something like a bridge certification authority.
In the diagram below, we see an example of such a system with two agents $b$ and $c$ we want to trust. On the left, we see the assumptions and on the right some trust derivations.

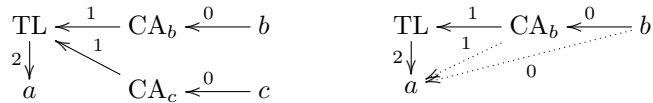
\[
\xymatrix@R=1em{
	\text{TL} \ar@{<-}^1[r] \ar@{<-}^1[rd] & \text{CA}_b \ar@{<-}^0[r] & b \\
	a \ar@{<-}^{2}[u] & \text{CA}_c \ar@{<-}^0[r] & c
}
\qquad 
\qquad
\xymatrix@R=1em{
	\text{TL} \ar@{<-}^1[r] & \text{CA}_b \ar@{<-}^0[r] & b  \\
	a \ar@{<-}^{2}[u] \ar@{<.}^1[ur] \ar@{<.}_0[urr] &  &
}
\]

\paragraph{Hierarchical PKI}
In one last example, we consider CAs linked in a hierarchical structure. Like in BCA, the Root CA (RCA) need only be trusted as a CA.
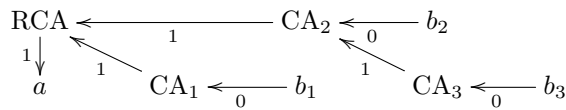
\[
\xymatrix@R=1em{
	\text{RCA} \ar@{<-}_{1}[rr] \ar@{<-}_{1}[rd] & & \text{CA}_2 \ar@{<-}_0[r]\ar@{<-}_{1}[rd] & b_2 & \\
	a \ar@{<-}^{1}[u] & \text{CA}_1 \ar@{<-}_0[r] & b_1 & \text{CA}_3 \ar@{<-}_0[r] & b_3
}
\]

\paragraph{Direct mutual trust} This is the goal of mutual trust.

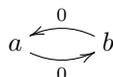
\[
\xymatrix{
	a \ar@{<-}@/^0.5pc/^0[r] & b \ar@/^0.5pc/@{<-}^0[l]
}
\]

\paragraph{Personal CA}
Here, everyone has their own CA to consult and check trustworthiness of other agents.
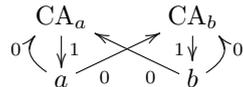
\[
\xymatrix@R=1em{
	\text{CA}_a \ar@{<-}@/_1pc/_0[d] & \text{CA}_b \ar@{<-}@/^1pc/^0[d] \\
	a \ar@{<-}_1[u]\ar@{->}_<<<<0[ur] & b \ar@{->}^<<<<0[ul]\ar@{<-}^1[u]
}
\]
We can derive $(a \overset{0}{\dashrightarrow} b)$ and $(b \overset{0}{\dashrightarrow} a)$.

\paragraph{Mesh PKI}
Here, everyone has their own CA they can use to advocate validity of their claims.
These CAs are connected in a mesh, where each CA trusts other CAs in their ability to authenticate keys.
A user then only needs to trust their own CA in being able to authenticate others' claims.
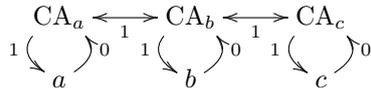
\[
\xymatrix@R=1em{
	\text{CA}_a \ar@{<->}_{1}[r] \ar@{<-}@/^1pc/^0[d] & \text{CA}_b \ar@{<->}_{1}[r]\ar@{<-}@/^1pc/^0[d] & \text{CA}_c \ar@{<-}@/^1pc/^0[d] \\
	a \ar@/^1pc/@{<-}^1[u] & b \ar@/^1pc/@{<-}^1[u]  & c \ar@/^1pc/@{<-}^1[u]
}
\]
Note that not all CAs need to be directly connected. As long as there is a chain of trust between each of the CAs, trust can be derived.

\section{Threshold Trust}\label{sec:threshold}

There are situations in which a CA is considered unreliable. This can happen for two reasons:
\begin{itemize}
	\item CAs may be incorrect, as in they are lying or simply wrong about the validity of a claim.
	\item CAs may be unresponsive regarding trust, as in they have not verified or communicated whether they consider someone trustworthy.\footnote{Note that the CA may still be responsive regarding forwarding of claims, as that is dealt with separately and can be independently verified.}
\end{itemize}
In both cases, we can use multiple CAs to mitigate the issue. 
In the first situation, we can confirm with multiple CAs whether an entity is trusted, before accepting the entity's claims.
In the second situation, we attempt to consult multiple CAs hoping at least some assert their trust in the entity.

Consider Lemma \ref{lem:comp-trust} again, and see that when $\N_S(a,i,b)$:
\[
\ntrust{1}{a}{i}{\tF} \wedge  \m{I}_{a \ot i}(\nval{0}{i}{b}{\tF}) \implies (\ntrust{0}{a}{b}{\tF})
\]
We see that trust composition requires two statements; the trust in the CA given by $\ntrust{0}{a}{i}{\tF}$, and the sharing of validation by the CA given by $\m{I}_{a \ot i}(\nval{0}{i}{b}{\tF})$.
If we think a CA is dishonest, we lack the former statement, and if we think a CA is unresponsive, we lack the latter statement.

Consider the following situation. We have two users $a$ and $b$, and two CAs named $i$ and $j$. We assume that trust of $a$ in $b$ gets delegated to both $i$ and $j$, asserting $\N_S(a,i,b)$ and $\N_S(a,j,b)$:
We get the following two results:
\begin{itemize}
	\item $\ntrust{1}{a}{i}{\tF} \vee \ntrust{1}{a}{j}{\tF}$ and $\m{I}_{a \ot i}(\nval{0}{i}{b}{\tF}) \wedge \m{I}_{a \ot j}(\nval{0}{j}{b}{\tF})$ implies $\ntrust{0}{a}{b}{\tF}$.
	If only one of the two CAs is trusted by $a$, then both need to assert trust in $b$ to guarantee that trust can be derived.
	\item $\ntrust{1}{a}{i}{\tF} \wedge \ntrust{1}{a}{j}{\tF}$ and $\m{I}_{a \ot i}(\nval{0}{i}{b}{\tF}) \vee \m{I}_{a \ot j}(\nval{0}{j}{b}{\tF})$ implies $\ntrust{0}{a}{b}{\tF}$.
	If both CAs are trusted, only one needs to assert trust in $b$.
\end{itemize}

The above two examples resolve problems created by potential attacks on the network. Possible incorrectness of a CA may be due to an attacker impersonating the CA, and in the first example trust can still be derived if one CA gets corrupted this way.
Possible unresponsiveness may be due to an attacker removing a CA's published certificates, and in the second example trust can still be derived if one CA gets disrupted. 

In cases where both incorrectness and unresponsiveness are possibilities, we would need to consult at least three CAs to validate a claim.
We get \emph{threshold PKI}, with the simplest being a \emph{2-out-of-3} threshold. Given formulas $A$, $B$, $C$, define: 
\[
\texttt{2of3}(A,B,C) := (A \wedge B) \vee (A \wedge C) \vee (B \wedge C)
\]
Note that $\texttt{2of3}(A,B,C) \equiv (A \vee B) \wedge (A \vee C) \wedge (B \vee C)$.

Consider again agents $a$ and $b$, together with three CAs $i$, $j$ and $k$, and let $S = \{aib,ajb,akb\}$, hence $\N_S(a,i,b)$, $\N_S(a,j,b)$, and $\N_S(a,k,b)$.
We get threshold trust:
\begin{lemma}
	If $\N_S(a,i,b)$, $\N_S(a,j,b)$, and $\N_S(a,k,b)$, then:
	
	\noindent
	$\texttt{2of3}(\ntrust{1}{a}{i}{\tF}, \ntrust{1}{a}{j}{\tF}, \ntrust{1}{a}{k}{\tF}) \wedge$
	
	\noindent
	$\texttt{2of3}(\m{I}_{a \ot i}\nval{0}{i}{b}{\tF}, \m{I}_{a \ot j}\nval{0}{j}{b}{\tF}, \m{I}_{a \ot k}\nval{0}{k}{b}{\tF})$
	
	$\implies \ntrust{0}{a}{b}{\tF}$
\end{lemma}

In~\cite{NordSec}, a logic is considered in which the trust thresholds are directly baked into the modalities themselves.

\section{Wish Modality}

We have seen how agents can request information. As discussed before, we have to be careful when using the communication modality, as any message is considered to be a claim made by the sender. As is, there is no direct way to formally communicate facts which are not yet believed without having to lie, except by forwarding other people's claims. This is quite inconvenient in networks based on trust.

A solution is to use a \emph{wish modality} $\mathcal{W}_a$, which tells us what an agent $a$ wants to be true.
The modality fits in the framework of this paper, since its satisfies axiom K; if agent $a$ desires $A$ to be true, and desires $B$ to be true, then one should be able to reason that $a$ wants both of them to be true. 
In cases where the agent wants either of the two to be true, but not necessarily both, we instead have to write $\mathcal{W}_a(A \vee B)$.

Most commonly, a wish can be used to send a request to another agent: $\mathcal{I}_{a,b}\mathcal{W}_bA$ means $b$ tells $a$ it wishes $A$ to be true. There are several examples of useful wishes one can utter:
\begin{itemize}
	\item With $\mathcal{I}_{a,b}\mathcal{W}_bt$, $b$ tells $a$ it likes some principal statement $t$ to be true, e.g. key ownership. If $a$ has the ability to make $t$ true, this is a direct request from $b$ to do so.
	\item With
	$\mathcal{I}_{a,b}\mathcal{W}_b\mathcal{I}_{c,a}A$, $b$ requests $a$ to send the claim $A$ to agent $c$. For instance, $a$ is some authority which is requested to send signed credentials to some retailer $c$.
	\item
	$\mathcal{I}_{a,b}\mathcal{W}_b\mathcal{B}_cA$ is a request of $b$ to $a$ to help convince $c$ of a certain fact.
\end{itemize}

As an example, suppose $t$ is the token expressing that $a$ owns a certain product or key.
Suppose we have a store $c$ which can make $t$ true, as specified by $\mathcal{B}_ct \Rightarrow t$.
A customer $a$ would like $t$ to be true and sends a request to $c$, $\mathcal{I}_{c,a}\mathcal{W}_at$.
Now, $c$ is willing to make $t$ true if $a$ has the proper credentials, which is expressed by some other token $u$. So $c$ sends $\mathcal{I}_{a,c}(u \Rightarrow t)$.

Now $a$ needs to convince $c$ of the truth of $u$. They can do so in several ways. If $c$ does not trust $a$ directly (at least not regarding $u$), we need an intermediate authority $b$ which $c$ trusts. Either $a$ sends a request to $b$ to validate and send the credentials to $c$ using the above request protocol.
Alternatively, $a$ can let $c$ know that $b$ can vouch for the truth of the credentials, and $c$ can request validation from $b$ independently. 

One last example is using the wish modality to request information. Suppose agent $b$ wants to know whether $A$ is true or not. In this case, it would make the statement $\m{W}_a((A \Rightarrow \mathcal{B}_aA) \wedge (\neg A \Rightarrow \mathcal{B}_a(\neg A)))$, and send this to someone who may know.
Of course, without prearranged obligations or guarantees, requests could be ignored. 

\section{Key Ownership and Proxy Trust}\label{sec:control}

One of the main motivations of this paper is to describe networks for authenticating ownership of keys. We could consider such keys as entities in their own right, used as aliases to send further messages. Even if one does not know the owner of a key, one may still reason about the belief and interactions of the owner, as represented by the key. In authorization logics, one may use statements of ownership to translate authority from an entity to their owned keys. Here we can similarly use key ownership here to translate trust in an entity to trust in their owned keys.

We take as set of tokens $\Tok = \{\m{C}_{a,b} \mid a, b \in \Ag\}$, where $\m{C}_{a,b}$ states $a$ is controlled by $b$. We consider this as expressing a reflexive and transitive relationship, with $\forall_a \m{C}_{a,a}$ and $\forall_{a,b,c} (\m{C}_{a,b} \Rightarrow \m{C}_{b,c} \Rightarrow \m{C}_{a,c})$ as additional public assumptions.
Unlike the delegation network $\RN$, ownership of keys is not public knowledge, and each entity would need to verify on their own whether a certain key is owned by a certain entity.

In public key infrastructures, entities make claims about ownership of keys. One can trust a person directly when they claim to control a key, or one may trust a CA to validate someones claims of ownership of keys.
This is modeled by instantiating the $\tF$ in trust statements $\ntrust{n}{a}{b}{\tF}$ with $\tF(c) = \m{C}_{d,c}$ for different choices of $d \in \Ag$.
This implements the claim: I have control of agent $d$.

Beside being examples of claims we want to validate, ownership of keys can have additional effects on the modalities. We can for instance add the following axioms to $\Delta$:
\begin{itemize}
	\item $\forall_{a,b}(\m{C}_{a,b} \Rightarrow \m{B}_b\m{C}_{a,b})$.
	\item $\forall_{a,b}(\m{C}_{a,b} \Rightarrow (\m{B}_a\X \Leftrightarrow \m{B}_b\X))$.
	\item $\forall_{a,b}(\m{C}_{a,b} \Rightarrow \forall_{c}(\m{I}_{a,c}\X \Rightarrow \m{I}_{b,c}\X))$.
	\item $\forall_{a,b}(\m{C}_{a,b} \Rightarrow \forall_{c}(\m{I}_{c,a}\X \Rightarrow \m{I}_{c,b}\X))$
\end{itemize}
These statements say the following; agents know which keys they control. If agent $a$ is controlled by agent $b$, then $a$ and $b$ believe the same things as they are identifiers (alter-egos) of the same entity. Moreover, messages sent to and from $a$ can be considered as sent to and from $b$ respectively.

Using the axioms above, we can \emph{transfer} trust.
If you trust an agent, then you can trust any agent controlled by them. This result will allow people to act using a digital identity they control, and still be trusted.

\begin{lemma}\label{lem:trust_push}
	$\vdash_\Delta  \forall_{b,c}(\m{C}_{b,c} \Rightarrow \forall_{a}(\m{V}_{a,c}\X \Rightarrow \m{V}_{a,b}\X))$.
\end{lemma}
\begin{proof}
	Suppose $\m{C}_{b,c}$ and $\m{V}_{a,c}A$ hold for some keys $a,b,c$ and statement $A$.
	We prove $\m{V}_{a,b}A$.
	
	If $\m{I}_{a,b}A$, then since $\m{C}_{b,c}$ we know by the new axiom that the message was actually sent by $c$, so $\m{I}_{a,c}A$. By validity $\m{V}_{a,c}A$, we know that $c$ believes it: $\m{B}_cA$. Again by $\m{C}_{b,c}$, we conclude that $b$ believes it too: $\m{B}_bA$.
	
	If $\m{B}_bA$, then by $\m{C}_{b,c}$ and the new axioms,  $\m{B}_cA$. Using reliability from $\m{V}_{a,c}A$ we can derive $A$. 
\end{proof}
\begin{corollary}\label{lem:trust_push2}
	$\vdash_\Delta  \forall_{a,b,c}\m{B}_a\m{C}_{b,c} \Rightarrow (\m{T}_{a,c}\X \Rightarrow \m{T}_{a,b}\X)$.
\end{corollary}
\begin{proof}
	Suppose $\m{B}_a\m{C}_{b,c}$ and $\m{T}_{a,c}A = \m{B}_a\m{V}_{a,c}A$ holds for some keys $a,b,c$ and statement $A$. Then $\m{T}_{a,b}A = \m{B}_a\m{V}_{a,b}A$ is true by lifting Lemma~\ref{lem:trust_push} to perspective $\m{B}_a$ using axiom (K).
\end{proof}

\section{Logic Calculi}\label{sec:calculi}

There are many different equivalent ways of formulating a calculus whose power matches our logic.
We highlight mainly three different calculi: 
\begin{enumerate}
	\item A basic intuitionistic sequent calculus for provability, as formulated in Figure \ref{Fig:basic}.
	\item A lambda calculus following the Fitch-style proof systems, as used in for instant dependent modal logics. \cite{Clouston18}
	\item A cut-free sequent calculus for decidable proof search.
\end{enumerate}
We shall briefly describe a version of the second, which is suitable for expressing proofs as terms in a lambda calculus. Note that the logic used in this paper is less powerful then those considered in recent works, e.g. as used in \emph{dependent modal type theory} in \cite{Birkedal,MDTT}. This is in order to facilitate easier proof searching, as is useful when imposing accountability of implied statements, as well as to not unnecessarily introduce unneeded complexity. The logic can of course be extended if needed, possibly sacrificing decidability.

We define \emph{modal contexts} as follows:
\[
\Gamma , \Delta := \varepsilon \mid \Gamma , x : A \mid \Gamma , \{\M\}
\]
These may contain assumptions in the form of statements $A$ with associated name $x$ for proofs.
Additionally, we have modalities $\{\M\}$ expressing a modal shift: e.g. the context $x : A, \{\B_a \}, y : B$ can be read as, supposing $A$ is true, and we consider the perspective of $a$, where we additionally suppose $a$ believes in $B$.

Let $[-]$ be the map from contexts to $\Mes^*$ defined as: $[\varepsilon] = \varepsilon$, $[\Gamma , x : A] = \Gamma$ and $[\Gamma , \{\M\}] = [\Gamma] , \M$.
\begin{figure}
	\[
	\frac{[\Delta] = \varepsilon}{\Gamma , x : A , \Delta \vdash \texttt{var}(x) : A}
	\qquad
	\frac{}{\Gamma \vdash * : \top}
	\qquad 
	\frac{\Gamma \vdash P : \bot}{\Gamma \vdash \texttt{abs}_A(P) : A}
	\qquad
	\frac{\Gamma , \{\M\} \vdash P : A}{\Gamma\vdash \texttt{lock}_{\M}(P) : \M A}
	\]
	\[
	\frac{\Gamma \vdash P : \M A \quad \M \Rrightarrow l \quad [\Delta] = l}{\Gamma , \Delta \vdash \texttt{key}_{\M \Rrightarrow l}(P) : A}
	\qquad
	\frac{\Gamma \vdash P : A \quad \Gamma \vdash Q : A \Rightarrow B}{\Gamma \vdash P \cdot Q : B}
	\qquad 
	\frac{\Gamma , x : A \vdash P : B}{\Gamma \vdash \lambda x:A.P : A \Rightarrow B}
	\]
	\[
	\frac{\Gamma \vdash P : A_1 \wedge A_2}{\Gamma \vdash \pi_i(P) : A_i}
	\qquad
	\frac{\Gamma \vdash P : A_1 \quad \Gamma \vdash Q : A_2}{\Gamma \vdash (P,Q) : A_1 \wedge A_2}
	\]
	\[
	\frac{\Gamma , x : A \vdash P : C \quad \Gamma , y : B \vdash Q : C \quad \Gamma \vdash R : A \vee B}{\Gamma \vdash \texttt{case}_C(R)\{\ip_1(x) \mapsto P, \ip_2(y) \mapsto Q\} : C}
	\qquad
	\frac{\Gamma \vdash P : A_i}{\Gamma \vdash \ip_i(P) : A_1 \vee A_2}
	\]
	\caption{Fitch-style lambda calculus for our logic}
	\label{fig:Fitch}
\end{figure}
See Figure \ref{fig:Fitch} for the Fitch-style lambda calculus.
Comparing this to the literature, here we rename \texttt{mod} to \texttt{lock}, and \texttt{unmod} to \texttt{key}. Furthermore, we do not apply our additional modal axioms on the contexts directly, instead folding those into the \texttt{key} operations. Hence, once one unlocks a term, the modal context gets fixed. Shifting modal contexts now becomes a meta operation on terms, similar to variable substitutions.

As an example, a possible lambda term for the sequent $x : \B_a(\I_{a \ot b} \p \Rightarrow \p) \vdash \I_{a \ot b} \p \Rightarrow \B_a \p$ is\\ $\lambda y : \I_{a \ot b} \p. \tLock{\B_a}{\App{\tKey{\B_a}{\B_a}{\Var{x}}}{\tLock{\I_{a \ot b}}{\tKey{\I_{a \ot b}}{\B_a \cc \I_{a \ot b}}{\Var{y}}}}}$.

\subsection{Meta constructions}

The intuitionistic modal logic forms a foundation for expressing a plethora of concepts.
We can construct a variety of other operations using the above tools.
For instance, with $\bot$ marking absurdity, we can define negation $\neg A$ as $\neg A = A \Rightarrow \bot$ following the standard intuitionistic tradition.

Another useful thing we can define is conjunction and disjunction over finite sets of formulas. Given a finite subset of $S \subseteq_{\text{fin}} \Form$, we define the following:
\begin{itemize}
	\item $\bigwedge S$ is the conjunction over $S$, which holds precisely when all $A \in S$ hold.
	\item $\bigvee S$ is the disjunction over $S$, which holds precisely if some $A \in S$ holds.
\end{itemize}
These are defined inductively, iterating binary conjunctions and disjunctions. In particular, $\bigwedge \emptyset \equiv \top$ and $\bigvee \emptyset \equiv \bot$.

When describing distributed networks, we consider the additional finite set of agents $\Ag$.
Both modalities and tokens may refer to specific agents, marking their perspectives and primitive statements of interest.
As such, formulas can refer to specific agents, and the referenced agent could be used as a parameter.
Given such a parametrized formula $f :  \Ag \to \Form$, and given a subset $S \subseteq \Ag$, we write: 
\begin{itemize}
	\item $\forall_{x \in S} f(x)$ for the universal quantification, equivalent to $\bigwedge \{f(a) \mid a \in S\}$.
	\item $\exists_{x \in S} f(x)$ for the existential quantification, equivalent to $\bigvee \{f(a) \mid a \in S\}$.
\end{itemize}
We will write $\forall_x f(x)$ and $\exists_x f(x)$ in case $S = \Ag$, and $\forall_{P(x)} f(x)$ and $\exists_{P(x)} f(x)$ if $S = \{a \in \Ag \mid P(a)\}$ with $P$ some predicate on agents.

These quantifiers could be expressed using formulas polymorphic over agent variables as well. However, we would like to avoid such unnecessary language generalisations at this point, focusing instead on the underlying theory of trust.

\subsection{Notes on Decidable Proof Search}

Let us briefly address how the decomposability property on axioms facilitates easier proof search.
Note that given decomposability, we can reduce any proof to one only using $\texttt{key}_{\M \Rrightarrow \N}$ and $\texttt{key}_{\M \Rrightarrow \N \cc \M \ominus \N}$ operators. Supposing we want to prove $\N B$ with a context containing $\M A$, then we know it is sufficient to use the formula in context as $\N A$ and/or $\N (\M \ominus \N) A$ depending which axioms are available.

More concretely, proof search may happen in a cut-free variant of the logic extending the usual calculi and their cut elimination proofs \cite{PFENNING}. These calculi use contexts excluding the $\{\M\}$ modal locks, instead defining for each modality $\N$ an explicit operation $\N^{-1}$ on contexts where $\N^{-1}(\Gamma)$ contains $A$ for any $\M A \in \Gamma$ such that $\M \Rrightarrow \N$, and $(\M \ominus \N)A$ for any $\M A \in \Gamma$ such that $\M \Rrightarrow \N \cc \M \ominus \N$.
We then handle modalities with the single rule:
\[
\frac{\M^{-1}(\Gamma) \vdash A}{\Gamma \vdash \M A}
\]
We then establish that the associated calculus has cut elimination, and note that we can associate a subformula order establishing that a proof search must terminate.
We get that if $\Mes$ is finite and $\mathsf{Ax}$ is a reflexive, transitive and decomposable set of axioms, the associated calculus has decidable proof search. Decidability may moreover be derived by proving our axioms as an instance of those systems specified an proven to be decidable in \cite{counter}, though this has not been verified.

This need not mean that the proof search is practical. Moreover, this is sensitive to the chosen axioms; note the need for decomposability and the lack of axioms of the form $\M \N A \Rightarrow \R A$ in our formalism.

\section{Kripke Model}\label{sec:model}

The proof system introduced in the previous section connects our logic to a \emph{categorical model} \cite{MDTT}, describing explicitly the handling and manipulation of proofs.
We shall furthermore consider a \emph{Kripke model} for establishing soundness and giving an interpretation of the formulas in terms of possible worlds which describe who knows and says what.

For the Kripke semantics, we consider a traditional variant \cite{Bozic1984-BOIMFN,Dosen1985ModelsFS} of intuitionistic modal logic, and also briefly consider Simpson's \cite{Simpson94a} version.

\begin{definition}
	A \emph{modal Kripke frame} on $(\Tok, \Mes)$ consists of a quadruple $(W,\leq , S_{-} , R_{-})$ where $W$ is a set of worlds, $\leq$ is a preorder on $W$, $S$ associates to each token $t \in \Tok$ a subset $S_t \subseteq W$ on $W$ and $R$ associates to each modality $\M \in \Mes$ a binary endorelation $R_\M \subseteq W^2$ on $W$ such that:
	\begin{enumerate}
		\item If $v \leq w$ and $v \in S_t$, then $w \in S_t$.
		\item If $v \leq w$ and $w R_\M w'$, then there is a $v' \in W$ such that $v R_\M v'$ and $v' \leq w'$. 
	\end{enumerate}
\end{definition}

At their core, Kripke frames consider possible worlds $W$, where each world not only describes which fundamental statements are true with subsets $S_t$, but also what is communicated and believed. The relation $R_{\B_a}$ for instance describes for each $v$ all possible worlds $w$ agent $a$ thinks they might be in. So if $v R_{\B_a} w \in S_t$, $v R_{\B_a} k \notin S_t$, then in world $v$ agent $a$ is unsure of whether statement $t$ is true, hence $\B_a t$ is not true in $v$.

Last but not least, $\leq$ implements an interpretation of intuitionistic logic. In this situation, statements are not simply true or false. Some statements either \emph{become} true or proven, and some statements \emph{remain} false or unproven. As such, we can think of $\leq$ as a progression of time; the more agents communicate with each other, the more they learn. This progression may be nondeterministic, as multiple distinct worlds may spawn from one world.

Note that in some definitions of Kripke frames, the second condition is removed and instead this property is baked into the interpretation of modal formulas. In our case, some conditions on $R$ cannot be avoided regardless, since they need to accommodate axioms like $\I_{a \ot b} A \Rightarrow \B_a \I_{a \ot b} A$.
Let $\mathsf{Ax}$ be a set of unfolding modal axioms on $\Mes$. 

\begin{definition}
	A modal Kripke frame $(W,\leq,S,R)$ conforms to $\mathsf{Ax}$ if
	$\M \Rrightarrow \N_1 \cc \dots \cc \N_n \in \mathsf{Ax}$ implies $R_{\N_1} ; \dots ; R_{\N_n} \subseteq R_\M$.
\end{definition}

Given a Kripke frame, we define \emph{satisfiability} of a formula $A$ in a world $w \in W$ inductively on formulas:
\begin{itemize}
	\item $w \models \top$, $w \not\models \bot$, and $w \models t$ iff $w \in S_t$
	\item $w \models \M A$ iff $\forall v. (w R_\M v \implies v \models A)$
	\item $w \models A \Rightarrow B$ iff $\forall v. ((w \leq v \wedge v \models A) \implies v \models B)$
	\item $w \models A \wedge B$ iff $w \models A$ and $w \models B$
	\item $w \models A \vee B$ iff $w \models A$ or $w \models B$
\end{itemize}

Some immediate properties:
\begin{itemize}
	\item If $w \models A$ and $w \leq v$ then $v \models A$.
	\item If the frame conforms to $\mathsf{Ax}$ and $\M \Rrightarrow \N_1 \cc \dots \cc \N_n \in \mathsf{Ax}$, then $w \models \M A \Rightarrow \N_1 \dots \N_n A$ for each $A \in \Form$ and $w \in W$.
\end{itemize}

\begin{theorem}
	If $\vdash A$ is provable, then $w \models A$ for any modal Kripke frame $(W,\leq,S,R)$ and world $w \in W´$.
\end{theorem}

\section{Conclusions}\label{sec:conclusions}

We end this paper with some final considerations.

\subsection{Regarding Privacy}

In this paper, we have not considered \emph{privacy}. We have considered responsiveness based on consent, for instance with the wish modality, where an authority has promised to share some fact after being asked to do so.
However, this does not prevent the same authority from sharing this fact without consent of the owner of the information.

Privacy specifically considers statements about entities \emph{not} believing or receiving messages about some sensitive piece of data.  Adding privacy may be an interesting next step, which would involve focusing more on negative statements, and would potentially necessitate adding further modalities.

\subsection{Universal Knowledge}

Both the modalities $\B_{\Omega}$ and $\Box$ describe a kind of public knowledge. The main difference is that we have $\Box X \Rightarrow X$, but not for $\B_{\Omega}$. 
For instance if we can prove $\B_{\Omega}(A_1 \wedge \dots \wedge A_n \Rightarrow B)$ we know that everyone (except us) can prove $A_1 , \dots , A_n \vdash B$. 

The $\Box$ can be used to make universally held assumptions. This particular use of the $\Box$ modality goes back to early proof theory by Gentzen \cite {Gentzen1935}, and is inspired by the work on \emph{dual-context modal logic} developed in \cite{Kavvos17}.
In the latter, a second context is used to describe formulas which are marked by the $\Box$ modality, and an equivalence is established between dual-context modal logic and modal logic with $\Box$.
We could further generalize the logic of this paper to use $\Box$ for marking public statements.
In order for this new modality to describe public knowledge, we need to add the following extra assumptions to the logic:
\begin{itemize}
	\item $\Box A \Rightarrow A$.
	\item $\Box A \Rightarrow \m{M}\Box A$ for any modality $\m{M}$.
\end{itemize}
As a result, $\Box$ satisfies the usual K4 axioms from modal logic.

Using the $\Box$ modality would allow one to reason about public statements internally, an agent may reason about consequences of releasing public statements.
However, it seems counter intuitive to think of public statements as anything but universally known, and hence we keep the fact whether a statement is public completely external in this paper. 

One could also generalise the public modality further.
It might be useful to instead specify different scopes of knowledge and groups of entities, using a preorder on modalities: $\m{M} \leq \m{N}$ if $\m{N}$ is aware of everything known by the perspective $\m{M}$. This can be expressed with axioms:
\begin{itemize}
	\item $\m{M}A \Rightarrow \m{N}A$
	\item $\m{M}A \Rightarrow \m{N}\m{M}A$.
\end{itemize}
An exploration of such systems is subject to future research.

\subsection{Evaluating Risk of Trust}

Trust may not always result in certainty. There is always a risk when building ones guarantees based on statements from other entities. As such, it is worthwhile to determine some risk related to a proof. Though this may be done by generalizing to probabilistic modal logic, a lot could already be achieved by simply analyzing proofs within the current logic. The fact of the matter is, that assuming sufficient trust, certain guarantees can be made. The risk is the trust assumptions themselves.

As such, one can collect different proofs of the same guarantee, and analyze the sets of trust assumptions necessary to make the proof. Then one could associate a certain level of risk to each assumption, depending on who needs to be trusted to get the result. This could be a specific risk (e.g. 5 percent), or some unspecified constant risk $\varepsilon$. For instance, in the 2-3 threshold example, associating a $5$ percent risk to trusting each of the three CAs gets us a total failure risk of only $0.725$ percent, which is a significantly reduced risk. 

\bibliographystyle{plain}
\bibliography{biblio}

\appendices 

\newpage

\section{Operational Semantics for Modal Lambda Calculus}
In this appendix, we expand on the semantics of our lambda calculus.
For $l \in \Mes^*$, we inductively define $\Gamma , \{l\}$ as:
\begin{itemize}
	\item $\Gamma , \{\} = \Gamma$,
	\item $\Gamma , \{\M , l\} = \Gamma , \{\M\} , \{l\}$.
\end{itemize}

First of all, we have weakening in two directions:
\begin{itemize}
	\item If $\Gamma , \Delta \vdash P : A$ then $\Gamma , x : A , \Delta \vdash P : A$.
	\item If $\Gamma \vdash P : A$ then $\{\M\} , \Gamma \vdash P : A$.
\end{itemize}

We write $l \Rrightarrow r_1 \mid \dots \mid r_n$ if $l = \M_1 \cc \dots \cc \M_n$ and $\M_i \Rrightarrow r_i$ for each $i$.
Given $[\Gamma] \Rrightarrow l_1 \mid \dots \mid l_n$ we can define the substitution $\Gamma\langle l_1 \mid \dots \mid l_n\rangle$ as: $\varepsilon\langle \rangle = \varepsilon$, and $(\Gamma , x : A)\langle l_1 \mid \dots \mid l_n\rangle = \Gamma\langle l_1 \mid \dots \mid l_n\rangle, x : A$, and $(\Gamma , \{\M\})\langle l_1 \mid \dots \mid l_n\rangle = \Gamma\langle l_1 \mid \dots \mid l_{n-1}\rangle , \{l_n\}$.
Then $[\Gamma\langle l_1 \mid \dots \mid l_n\rangle] = l_1,\dots,l_n$.
This modal substitution can be extended to terms as well, allowing for context shifting.
\begin{lemma}
	If $[\Gamma] \Rrightarrow l_1 \mid \dots \mid l_n$ and $\Gamma \vdash P : A$ then $\Gamma\langle l_1 \mid \dots \mid l_n\rangle \vdash P\langle l_1 \mid \dots \mid l_n\rangle : A$ for some $P\langle l_1 \mid \dots \mid l_n\rangle$.
\end{lemma}
\begin{proof}
	We define $P\langle l_1 \mid \dots \mid l_n\rangle$ by induction:
	\begin{itemize}
		\item $\texttt{var}(x)\langle l_1 \mid \dots \mid l_n\rangle = \texttt{var}(x)$.
		\item $*\langle l_1 \mid \dots \mid l_n\rangle = *$.
		\item $\texttt{abs}_A(P)\langle l_1 \mid \dots \mid l_n\rangle = \texttt{abs}_A(P\langle l_1 \mid \dots \mid l_n\rangle)$.
		\item 
		$\tLock{\M}{P}\langle l_1 \mid \dots \mid l_n \rangle = \tLock{\M}{P\langle l_1 \mid \dots \mid l_n \mid \M \rangle}$.
		\item Suppose $\Gamma \vdash P : \M A$, $\M \Rrightarrow r$, $[\Delta] = r$ and $[\Gamma , \Delta] \Rrightarrow l_1 \mid \dots \mid l_n$.
		We can divide the last statement into $[\Gamma] \Rrightarrow l_1 \mid \dots \mid l_i$ and $[\Delta] = r \Rrightarrow l_{i+1} \mid \dots \mid l_n$.
		By transitivity, $\M \Rrightarrow l_{i+1} , \dots , l_n$, so we can define
		$\tKey{\M}{r}{P}\langle l_1 \mid \dots \mid l_n \rangle = \tKey{\M}{l_{i+1} \cc \dots \cc l_n}{P\langle l_1 \mid \dots \mid l_i \rangle}$.
	\end{itemize}
\end{proof}
We can also perform partial modal substitutions. For instance, if $[\Gamma] = \M_1,\dots,\M_n,[\Delta]$ and $[\Delta] \Rrightarrow l_1 \mid \dots \mid l_m$, then we can write $(-)\langle l_1 \mid \dots \mid l_n \rangle$ as applied to $\Gamma$ and terms in context $\Gamma$ as a shorthand for $(-)\langle \M_1 \mid \dots \mid \M_n \mid l_1 \mid \dots \mid l_n \rangle$.

Similarly, we can define variable substitution. Given $\Gamma , x : A , \Delta \vdash P : B$ and $\Gamma \vdash Q : A$ there is a $P[Q / x]$ such that $\Gamma , \Delta \vdash P[Q / x] : B$.
This implements a cut rule.
We get the following normalisation procedure, which tidies up the terms and hence the proofs.
\begin{itemize}
	\item $\tKey{\M}{l}{\tLock{\M}{P}} \rightsquigarrow P\langle l \rangle$.
	\item $(\lambda x : A.P) \cdot Q \rightsquigarrow P[Q / x]$.
	\item $\pi_i((P_1 , P_2)) \rightsquigarrow P_i$.
	\item $\texttt{case}_C(\ip_i(R))\{\ip_1(x_1) \mapsto P_1, \ip_2(x_2) \mapsto P_2\} \rightsquigarrow P_i[R/x_i]$.
	\item $\texttt{abs}_A(\texttt{abs}_\bot(P)) \rightsquigarrow \texttt{abs}_A(P)$.
	\item $\texttt{abs}_{A \Rightarrow B}(P) \cdot Q \rightsquigarrow \texttt{abs}_{B}(P)$.
	\item $\pi_i(\texttt{abs}_{A_1 \wedge A_2}(P)) \rightsquigarrow \texttt{abs}_{A_i}(P)$.
	\item $\texttt{case}_C(\texttt{abs}_{A \vee B}(R)\{\dots\}) \rightsquigarrow \texttt{abs}_C(R)$.
\end{itemize}

\section{Proof of Decidability}

Suppose we have a finite set of modalities $\Mes$ and a reflexive, transitive and decomposable set of axioms $\mathsf{Ax}$.
We shall write $\M \Rrightarrow l$ to mean that $\M \Rrightarrow l \in \mathsf{Ax}$.
We take the partial map $\ominus : \Mes^2 \rightharpoonup \Mes$, where $\M \ominus \N$ is defined as a witness to the decomposability property. So $\M \Rrightarrow \N \cc (\M \ominus \N)$ if and only if $\M \ominus \N$ is defined, and if $\M \Rrightarrow \N \cc l$ with $l$ having at least one element, then $\M \ominus \N$ is defined and $\M \ominus \N \Rrightarrow l$.
We write $\M \ominus \N \downarrow$ if it is defined, and $\M \ominus \N \uparrow$ if it is not.

A \emph{flat context} $\Gamma$ is given by a list of formulas. In other words, it is a context without any $\{\M \}$ in it.
For any modality $\M$, we define a function $(\_) \ominus \M$ on flat contexts (which is written as $\M^{-1}(-)$ in the main body of the paper), where:
\begin{itemize}
	\item $() \ominus \N = ()$.
	\item $(\Gamma, \M B) \ominus \N =$ \\
	$\begin{cases}
		\Gamma \ominus \N , B , (\M \ominus \N)B & \text{if } \M \Rrightarrow \N, \text{ and } \M \ominus \N \downarrow\\
		\Gamma \ominus \N , B  & \text{if } \M \Rrightarrow \N, \text{ and } \M \ominus \N \uparrow \\
		\Gamma \ominus \N , (\M \ominus \N)B & \text{if } \neg \M \Rrightarrow \N, \text{ and } \M \ominus \N \downarrow \\
		\Gamma \ominus \N  & \text{if } \neg \M \Rrightarrow \N, \text{ and } \M \ominus \N \uparrow\\
	\end{cases}$
	\item $(\Gamma, B) \ominus \N = \Gamma \ominus \N$ for any other $B$.
\end{itemize}

We write $\Gamma \subseteq \Delta$ if any formula of $\Gamma$ is in $\Delta$. We write $\Gamma \sqsubseteq \Delta$ if any formula of $\Gamma$ is either in $\Delta$, or of the form $\M C$ with $\Delta$ containing $\N C$ for some $\N$ such that $\N \Rrightarrow \M$.
This is capturing the fact that a formula $\N C$ is more general then $\M C$ if $\N \Rrightarrow \M$. Hence in $\Gamma \sqsubseteq \Delta$, $\Delta$ is a stronger set of assumptions than $\Gamma$. By transitivity of $\mathsf{Ax}$, $\sqsubseteq$ is transitive.

\begin{lemma}\label{lem:vk-shift}
	If $\Gamma \sqsubseteq \Delta$, then $\Gamma \ominus \M \sqsubseteq \Delta \ominus \M$
\end{lemma}

\begin{proof}
	For $C \in (\Gamma \ominus \M)$, either $C = C'$, $\N C' \in \Gamma$ with $\N \Rrightarrow \M$, or $C = (\N \ominus \M)C'$ and $\N C' \in \Gamma$. Regardless of the case, $\Delta$ has $\R C'$ such that $\R \Rrightarrow \N$ (note that $\R$ could be $\N$).
	\begin{itemize}
		\item If $C = C'$, and $\N C' \in \Gamma$ with $\N \Rrightarrow \M$. By transitivity, $\R \Rrightarrow \M$, hence $C = C' \in (\Delta \ominus \M)$.
		\item If $C = (\N \ominus \M)C'$, and $\N C' \in \Gamma$, then $\R \Rrightarrow \M \cc (\N \ominus \M)$, hence $\R \ominus \M$ is defined and $\R \ominus \M \Rrightarrow \N \ominus \M$ by definition. Hence $(\R \ominus \M)C' \in (\Delta \ominus \M)$, which covers for $C = (\N \ominus \M)C'$.
	\end{itemize}
\end{proof}

\begin{lemma}\label{lem:rel-shift}
	If $\M \Rrightarrow \N$, then $(\Gamma \ominus \M) \sqsubseteq (\Gamma \ominus \N)$.
\end{lemma}

\begin{proof}
	For $C \in (\Gamma \ominus \M)$, then we have two cases:
	\begin{itemize}
		\item If $C = C'$, $\R C' \in \Gamma$ with $\R \Rrightarrow \M$. By transitivity, $\R \Rrightarrow \N$, hence $C = C' \in (\Gamma \ominus \N)$.
		\item If $C = (\R \ominus \M)C'$ and $\R C' \in \Gamma$. By transitivity, $\R \Rrightarrow \N \cc (\R \ominus \M)$, hence $(\R \ominus \N)$ is defined and $(\R \ominus \N) \Rrightarrow (\R \ominus \M)$. So $(\R \ominus \N)C' \in (\Gamma \ominus \N)$ which covers for $C = (\R \ominus \M)C'$.
	\end{itemize}
\end{proof}

\begin{lemma}\label{lem:split-inc}
	If $\M \ominus \N$ is defined, then $(\Gamma \ominus \M) \sqsubseteq ((\Gamma \ominus \N) \ominus (\M \ominus \N))$.
\end{lemma}

\begin{proof}
	For $C \in (\Gamma \ominus \M)$, then we have two cases:
	\begin{itemize}
		\item If $C = C'$, $\R C' \in \Gamma$ with $\R \Rrightarrow \M$. By transitivity, $\R \Rrightarrow \N \cc \M \ominus \N$, 
		hence $\R \ominus \N$ exists and $(\R \ominus \N) \Rrightarrow (\M \ominus \N)$. So, $(\R \ominus \N)C' \in ((\Gamma \ominus \N)$, and $C' \in ((\Gamma \ominus \N) \ominus (\M \ominus \N))$.
		\item If $C = (\R \ominus \M)C'$ and $\R C' \in \Gamma$. By transitivity, $\R \Rrightarrow \N \cc \M \ominus \N \cc \R \ominus \M$.
		Hence $\R \ominus \N$ exists and $(\R \ominus \N) \Rrightarrow \M \ominus \N \cc \R \ominus \M$. Hence $(\R \ominus \N) \ominus (\M \ominus \N)$ exists, and $(\R \ominus \N) \ominus (\M \ominus \N) \Rrightarrow \R \ominus \M$.
		So $(\R \ominus \N)C' \in (\Gamma \ominus \N)$ and $((\R \ominus \N) \ominus (\M \ominus \N))C' \in ((\Gamma \ominus \N) \ominus (\M \ominus \N))$ which covers for $C = (\R \ominus \M)C'$.
	\end{itemize}
\end{proof}

\subsection{The sequent calculus}

\begin{figure}
	\[
	\frac{
	}{
		\Gamma , t , \Delta \vdash t
	}(Var)
	\qquad
	\frac{
		\Gamma , A \vdash B
	}{
		\Gamma \vdash A \Rightarrow B
	}(ImpR)
	\qquad
	\frac{
		\Gamma , A \Rightarrow B \vdash A \qquad \Gamma , A \Rightarrow B , B \vdash D
	}{
		\Gamma , A \Rightarrow B \vdash D
	}(ImpL)
	\]
	\[
	\frac{
		\Gamma \ominus \M \vdash A
	}{
		\Gamma \vdash \M A
	}(ModR)
	\qquad
	\frac{}{\Gamma \vdash \top}(TopR)
	\qquad 
	\frac{}{\Gamma , \bot , \Delta  \vdash A}(BotL)
	\]
	\[
	\frac{
		\Gamma , A_1 \wedge A_2 , A_1   \vdash B
	}{
		\Gamma , A_1 \wedge A_2   \vdash B
	}(AndL1)
	\quad
	\frac{
		\Gamma , A_1 \wedge A_2 , A_2   \vdash B
	}{
		\Gamma , A_1 \wedge A_2  \vdash B
	}(AndL2)
	\qquad
	\frac{
		\Gamma \vdash B_1 \qquad \Gamma \vdash B_2
	}{
		\Gamma \vdash B_1 \wedge B_2
	}(AndR)
	\]
	\[
	\frac{
		\Gamma , A_1 \vee A_2 , A_1  \vdash B \qquad \Gamma , A_1 \vee A_2 , A_2  \vdash B 
	}{
		\Gamma , A_1 \vee A_2   \vdash B
	}(OrL)
	\qquad
	\frac{
		\Gamma \vdash B_1
	}{
		\Gamma \vdash B_1 \vee B_2
	}(OrR1)
	\qquad
	\frac{
		\Gamma \vdash B_2
	}{
		\Gamma \vdash B_1 \vee B_2
	}(OrR2)
	\]
	\caption{Decidable Sequent Calculus}
	\label{Fig:seq}
\end{figure}

Figure \ref{Fig:seq} gives the sequent calculus, showing proof rules to determine which sequent of the form $\Gamma \vdash A$, with $\Gamma$ a flat context, are provable.

\begin{proposition}[Structural weakening]\label{prop:struct}
	Suppose $\D$ gives a proof of $\Gamma \vdash C$, and $\Gamma \sqsubseteq \Delta$, then there is a proof $\D'$ of $\Delta \vdash C$, where $\D'$ has the same shape of $\D$.
\end{proposition}
\begin{proof}
	Can be done by induction on $\D$. In the ModR rule, we use Lemma \ref{lem:vk-shift} in order to make the inductive call. Note in particular that the (Var) rule only targets tokens, and if $\Gamma$ has token $t$, then $\Delta$ has token $t$ as well.
\end{proof}

We shall freely apply the above lemma to modify contexts accordingly, including swapping and copying formulas.

\begin{proposition}[Identity theorem]
	For any formula $A$, the sequent $\Gamma, A \vdash A$ is provable.
\end{proposition}
\begin{proof}
	Proven by induction on $A$:
	\begin{itemize}
		\item If $A = t$, $A = \bot$ or $A = \top$, it is directly provable by (Var), (BotL) and (TopR) respectively.
		\item If $A = A_1 \Rightarrow A_2$:
		
		\scalebox{.8}{
			\AxiomC{}
			\RightLabel{(IH)}
			\UnaryInfC{$\Gamma, A_1 \Rightarrow A_2 , A_1 \vdash A_1$}
			\AxiomC{}
			\RightLabel{(IH)}
			\UnaryInfC{$\Gamma, A_1 \Rightarrow A_2 , A_1 , A_2 \vdash A_2$}
			\RightLabel{(ImpL)}
			\BinaryInfC{$\Gamma, A_1 \Rightarrow A_2, A_1 \vdash A_2$}
			\RightLabel{(ImpR)}
			\UnaryInfC{$\Gamma, A_1 \Rightarrow A_2 \vdash A_1 \Rightarrow A_2$}
			\DisplayProof
		}
		\item If $A = \M B$,
		\begin{center}
			\AxiomC{}
			\RightLabel{(IH)}
			\UnaryInfC{$\Gamma \ominus \M , B , (\M \ominus \M)B \vdash B$}
			\RightLabel{(Mod)}
			\UnaryInfC{$\Gamma, \M B \vdash \M B$}
			\DisplayProof
		\end{center}
		Leave out $(\M \ominus \M)B$ if $(\M \ominus \M)$ is undefined. $\M \Rrightarrow \M$ holds by reflexivity, hence $(\Gamma , \M B) \ominus \M$ contains $B$.
		\item The $\wedge$ and $\vee$ cases are standard.
	\end{itemize}
\end{proof}

\begin{proposition}[Cut elimination]
	Given a proof $\D$ of $\Gamma \vdash A$, and a proof $\E$ of $\Gamma, A \vdash B$, then we can construct a proof $\F$ of $\Gamma \vdash B$.
\end{proposition}
\begin{proof}
	We use Pfenning's structural cut elimination proof \cite{PFENNING} as a basis.
	We do induction on: Size of $A$, size of $\E$, size of $\D$, in that order.
	We consider the size of $\M A$ and $\N A$ for any two modalities $\M$ and $\N$ to be the same.

	We focus on the non-standard new case our calculus includes: Both $\D$ and $\E$ end with ModR.
	
	Case. 
	
	\begin{center}
		$\D =$
		\AxiomC{$\D_1$}
		\RightLabel{(d)}
		\UnaryInfC{$\Gamma \ominus \M \vdash A$}
		\RightLabel{(ModR)}
		\UnaryInfC{$\Gamma \vdash \M A$}
		\DisplayProof
		\qquad
	\end{center} 

	\begin{center}
		$\E =$
		\AxiomC{$\E_1$}
		\RightLabel{(e1)}
		\UnaryInfC{$\Gamma \ominus \N , A , (\M \ominus \N)A \vdash B$}
		\RightLabel{(ModR)}
		\UnaryInfC{$\Gamma , \M A \vdash \N B$}
		\DisplayProof
	\end{center}
	Note that the $(\M \ominus \N)A$ and $A$ under $\E_1$ exist depending on whether $\M \ominus \N$ exists and $\M \Rrightarrow \N$ holds. We shall inductively cut these two formulas. If the formula to be cut does not exist, the respective cut can be left out.
	
	To cut $(\M \ominus \N)A$, we use Lemma \ref{lem:split-inc} to note that $(\Gamma \ominus \M) \sqsubseteq ((\Gamma \ominus \N) \ominus (\M \ominus \N))$, and by Lemma \ref{lem:vk-shift}, $((\Gamma \ominus \N) \ominus (\M \ominus \N)) \sqsubseteq ((\Gamma \ominus \N , A) \ominus (\M \ominus \N))$ (regardless of whether this $A$ is there), hence by structural weakening $\D_1$ gives a proof of $(\Gamma \ominus \N , A) \ominus (\M \ominus \N) \vdash A$.
	
	To cut $A$, then since $A$ is there by the prerequisite that $\M \Rrightarrow \N$, we can use Lemma \ref{lem:rel-shift} to see that $(\Gamma \ominus \M) \sqsubseteq (\Gamma \ominus \N)$. Hence by structural theorem, $\D_1$ gives a proof of $\Gamma \ominus \N \vdash A$.
	
	We use these versions of $\D_1$ to perform the cuts in the following way, $\F =$
	
	\scalebox{1}{
		\AxiomC{$\D_1'$}
		\RightLabel{}
		\UnaryInfC{$\Gamma \ominus \N \vdash A$}
		\AxiomC{$\D_1'$}
		\UnaryInfC{$(\Gamma \ominus \N , A) \ominus (\M \ominus \N) \vdash A$}
		\RightLabel{}
		\UnaryInfC{$\Gamma \ominus \N , A \vdash (\M \ominus \N)A$}
		\AxiomC{$\E_1$}
		\UnaryInfC{$\Gamma \ominus \N , A , (\M \ominus \N)A \vdash B$}
		\RightLabel{(IH-$\E$)}
		\BinaryInfC{$\Gamma \ominus \N , A \vdash B$}
		\RightLabel{(IH-$A$)}
		\BinaryInfC{$\Gamma \ominus \N \vdash B$}
		\RightLabel{}
		\UnaryInfC{$\Gamma \vdash \N B$}
		\DisplayProof
	}
	
\end{proof}

It should be noted that in the above proof, the fact that \textsf{Ax} is decomposable is used in a fundamental way. $\Gamma \ominus \N$ can for each $\M A$ only spin up $A$ and $(\M \ominus \N) A$ at most. We can freely cut the $A$, as a reduced formula ($A$ compared to $\M A$) can be used as a basis for the induction hypothesis, allowing us to liberally change the needed proofs as the size of formula is the dominant (first) inductive argument in the cut elimination proof. Cutting $(\M \ominus \N)A$ is more difficult: we ensure that for determining termination of the cut elimination proof, we consider the size of $(\M \ominus \N)A$ to be the same as the size of $\M A$. This is ok, as they both have the same number of \emph{proper} subformulas, and hence the measure for checking termination on the formula to be cut does not increase. Cutting $\M \ominus \N)A$ can be done using a reduced proof $\E$ and same size proof $\D$. 

Decomposability also ensures that we need only consider one such modal formula of the same size as $\M A$. Cutting multiple is problematic, as it makes the termination argument more difficult (maybe impossible), so we let $(\M \ominus \N)A$ be sufficient for implying all other modal formulas $\R A$ which may be needed. So if we have multiple $\R$ for which $\M \Rrightarrow \N \cdot \R$, the single formula $(\M \ominus \N)A$ covers all of them.

\subsection{Proving Decidability}

The decidability proof follows the usual recipe.
Consider again the subformula property $\leq$ of formulas, which is extended to include $A \leq \M A$ and $\M A \leq \N B$ if $A \leq B$.
Since we only have a finite number of modalities $\Mes$, we get that for any formula $A$ there are only finitely many $B$ such that $B \leq A$.
We observe that for any $A \in (\Gamma \ominus \M)$ there is a $B \in \Gamma$ such that $A \leq B$.
As a consequence, we can show by induction on proofs that suppose $\D$ proves $\Gamma \vdash A$, and $\Delta \vdash B$ appears somewhere in $\D$, then for any $C \in \Delta, B$ there is a $D \in \Gamma , A$ such that $C \leq D$.

Given structural weakening shown in Proposition \ref{prop:struct}, we see that whenever a $\Delta \vdash A$ appears above a $\Gamma \vdash A$ in a proof, where $\Delta \subseteq \Gamma$, then we have a redundancy: the proof of $\Delta \vdash A$ could replace the proof of $\Gamma \vdash A$, reducing the proof tree. We say that $\D$ is reduced if there are no $\Delta \vdash A$ appearing above $\Gamma \vdash A$ in $\D$ such that $\Delta \subseteq \Gamma$.
By the above remark, if there is a proof of a sequent, there is a reduced proof of that sequent.

For each sequent $\Gamma \vdash A$, there is a bound $n$ such that any reduced proof of $\Gamma \vdash A$ has at most height $n$. This is because by the subformula property, any $\Delta \vdash B$ in a proof of $\Gamma \vdash A$ can only contain formulas drawn from the set of subformulas of $\Gamma , A$, of which there are finitely many (note there are only finitely many modalities): we say there are $m$ such formulas. So suppose $S$ is a set of sequents $\Delta \vdash B$ possibly appearing in proofs of $\Gamma \vdash A$, such that no two sequents are eachother weakenings, then $S$ has at most $m \cdot 2^m$ elements ($m$ possible consequents, and $2^m$ possible ordered lists of non-repeating assumptions).
So any branch in a reduced proof has at most $m \cdot 2^m$ giving a bound on the height of reduced trees.

Since each proof steps makes at most two branches, the bound on the height gives a bound on the size of reduced proofs. Hence, up to weakening, there are only a finite number of reduced proof trees we need to check to see if $\Gamma \vdash A$ has a proof. For a theoretical argument of decidability, we simply check all possible reduced proofs to find if one works. More practical algorithms involve more targeted proof searches.

\subsection{Extra: Proving axioms}
We can show that the sequent calculus correctly adapts $\mathcal{L}_{\mathsf{Ax}}$, as it has the same rules as the standard sequent calculus for intuitionistic logics on the non-modal side.
Let us prove the axioms:
\begin{itemize}
	\item For axiom K, note that $(\M A , \M (A \Rightarrow B)) \ominus \M$ by reflexivity of \textsf{Ax} includes $A$ and $A \Rightarrow B$, which proves $B$. So $(\M A , \M (A \Rightarrow B)) \ominus \M \vdash B$ is provable, and hence by the ModR rule, $\M A , \M (A \Rightarrow B) \vdash \M B$. Axiom K is then constructed by using the ImpR rule twice.
	\item For necessity, note that with have a proof of $\vdash A$, then since $(\cdot \ominus \M) = \cdot$, this is also a proof of $(\cdot \ominus \M) \vdash A$. Hence by the ModR rule, $\vdash \M A$ is provable.
	\item Suppose $\M \Rrightarrow \N \in \textsf{Ax}$, then $(\M A \ominus \N)$ has $A$.
	So $(\M A \ominus \N) \vdash A$ is provable by the identity theorem, and hence $\M A \vdash \N A$ is provable by ModR, and $\vdash \M A \Rightarrow \N A$ by ImpR.
	\item Suppose $\M \Rrightarrow \N \cdot \R \in \textsf{Ax}$, then $(\M A \ominus \N)$ has $(\M \ominus \N)A$ and $\M \ominus \N \Rrightarrow \R \in \textsf{Ax}$ by decomposability.
	So $(\M A \ominus \N) \ominus \R$ has $A$, and we can prove $(\M A \ominus \N) \ominus \R \vdash A$, and apply ModR twice to get $\M A \vdash \N \R A$, proving $\vdash \M A \Rightarrow \N \R A$ with ImpR.
	\item We generalize the previous case, showing that by induction on the length of $l$, if $\M \Rrightarrow l \in \textsf{Ax}$ and $\M A \in \Gamma$, then $\Gamma \vdash l(A)$.
	Base case for length of $l$ being 1 (or even 2) has been covered above.
	
	Suppose $l = \N , l'$, then $\M \ominus \N$ exists and $\M \ominus \N \Rrightarrow l' \in \textsf{Ax}$. So supposing $\M A \in \Gamma$, then $(\M \ominus \N)A \in (\Gamma \ominus \N)$, and by induction hypothesis, $(\Gamma \ominus \N) \vdash l'(A)$. Using ModR, we can conclude that $\Gamma \vdash \N l'(A)$ where $\N l' = l$.
	This finishes the induction.
	
	We conclude that if $\M \Rrightarrow l \in \textsf{Ax}$, then $\M A \vdash l(A)$ and hence $\vdash \M A \Rightarrow l(A)$ by ImpR.
\end{itemize}

\section{Equivalence between Calculi}

Given two flat contexts $\Gamma$ and $\Delta$, we say that $\Gamma \vdash \Delta$ is provable if $\Gamma \vdash B$ is provable for any $B \in \Delta$. 

\begin{lemma}
	If $\Gamma \vdash \Delta$ and $\Delta \vdash \Phi$ are provable, then $\Gamma \vdash \Phi$ is provable.
\end{lemma}

\begin{proof}
	Starting with $\Gamma, \Delta \vdash C$ for some $C \in \Phi$, use a sequence of sequents $\Gamma , B_1, \dots , B_{i-1} \vdash B_i$ to cut away $\Delta$, until getting to $\Gamma \vdash C$.
\end{proof}

Consider a modal context $\Gamma$, which may include modalities $\{\M\}$. Hence $\Gamma$ is either flat (has no modalities), or of the form $\Gamma' , \{ \M \} , \Gamma''$ with $\Gamma'$ flat and $\Gamma'$ a modal context.

\begin{definition}
	Given context $\Gamma$ and flat context $\Delta$, then $\Gamma \vdash \Delta$ is \emph{constructable} if, by induction on $\Gamma$:
	\begin{itemize}
		\item If $\Gamma$ is flat, then $\Gamma \vdash \Delta$ is constructable if it is provable.
		\item If $\Gamma = \Gamma_0 , \{\M\} , \Gamma_1$ with $\Gamma_0$ flat, then $\Gamma \vdash \Delta$ is constructable if there is a flat context $\Phi$ such that:
		\begin{itemize}
			\item $\Gamma_0 \vdash \Phi$ is provable.
			\item $\Phi \ominus \M , \Gamma_1 \vdash \Delta$ is constructable.
		\end{itemize}
	\end{itemize}
\end{definition}

If $\Gamma \vdash \Phi$ is constructable, and $\Phi, \Psi \vdash \Delta$ is constructable, then $\Gamma , \Psi \vdash \Delta$ is constructible.

\begin{lemma}[Weakening of Constructability]
	If $\Phi \subseteq \Psi$ and $\Gamma , \Phi , \Delta \vdash \Omega$ is constructable, then $\Gamma , \Psi , \Delta \vdash \Omega$ is constructable.
\end{lemma}

\begin{lemma}\label{lem:con-merge}
	If $\Gamma \vdash \Delta_0$ and $\Gamma \vdash \Delta_1$ are constructable, then $\Gamma \vdash \Delta_0 , \Delta_1$ is constructable.
\end{lemma}

\begin{proof}
	Prove is done by induction on $\Gamma$.
	
	If $\Gamma $ is flat, then $\Gamma \vdash \Delta_0$ and $\Gamma \vdash \Delta_1$ are provable, hence $\Gamma \vdash \Delta_0, \Delta_1$ is provable.
	
	If $\Gamma = \Gamma_0 , \{\M\} , \Gamma_1$ with $\Gamma_0$ flat, then there are $\Phi_0$ and $\Phi_1$ such that $\Gamma_0 \vdash \Phi_0$ and $\Gamma_0 \vdash \Phi_1$ are provable and hence $\Gamma_0 \vdash \Phi_0, \Phi_1$ is provable, and both $\Phi_0 \ominus \M , \Gamma_1 \vdash \Delta_0$ and $\Phi_1 \ominus \M , \Gamma_1 \vdash \Delta_1$ are constructive. Note that $(\Phi_0 , \Phi_1) \ominus \M = (\Phi_0 \ominus \M) , (\Phi_1 \ominus \M)$, hence by weakening $(\Phi_0 , \Phi_1) \ominus \M , \Gamma_1 \vdash \Delta_0$ and $(\Phi_0 , \Phi_1) \ominus \M , \Gamma_1 \vdash \Delta_1$ are constructive. By induction hypothesis, $(\Phi_0 , \Phi_1) \ominus \M , \Gamma_1 \vdash \Delta_0, \Delta_1$.
\end{proof}

\begin{theorem}
	If there is a term $\Gamma \vdash P : A$, then $\Gamma \vdash A$ is constructable.
\end{theorem}

\begin{proof}
	Induction on $t$:
	
	$\App{t}{r}$ with $t : A \Rightarrow B$ and $r : A$, then by induction hypothesis, $\Gamma \vdash A \Rightarrow B$ and $\Gamma \vdash A$ are constructable, hence by Lemma \ref{lem:con-merge}, $\Gamma \vdash A \Rightarrow B , A$ is constructable. We show that $A \Rightarrow B , A \vdash B$ is provable, to construct $\Gamma \vdash B$.
	\begin{center}
		\AxiomC{}
		\RightLabel{(Var)}
		\UnaryInfC{$A \Rightarrow B, A \vdash A $}
		\AxiomC{}
		\RightLabel{(Var)}
		\UnaryInfC{$A \Rightarrow B , A, B \vdash B$}
		\RightLabel{(ImpL)}
		\BinaryInfC{$ A \Rightarrow B , A \vdash B$}
		\DisplayProof
	\end{center}
	
	$\Lam{A}{t}$ with $t : B$, then by induction hypothesis $\Gamma , A \vdash B$ is constructable. Take the tail-end provable $\Gamma , A \vdash B$ and apply (ImpR) to construct $\Gamma \vdash A \Rightarrow B$.
	
	$\tKey{\M}{\alpha}{t}$ with $\Gamma \vdash t : \M A$. By induction hypothesis, $\Gamma \vdash \M A$ is constructable. We do case distinction on $\alpha$:
	\begin{itemize}
		\item If $\alpha = \N$ then $\M \Rrightarrow \N$. $\M A , \{\M\} \vdash A$ is constructible since $\M A \vdash \M A$ and $\M A \ominus \N \vdash A$ is provable by (Var). Composing $\M A , \{\N\} \vdash A$ with $\Gamma , \{\N\} \vdash \M A$, we construct , $\{\N\} , \Gamma \vdash A$.
		\item If $\alpha = \N_1, \dots , \N_n$ with $n > 1$, we know $\M \Rrightarrow \N_1 \cc \dots \cc \N_n$, and hence can define $\R_1 = \M$, $\R_{i+1} = \R_i \ominus \N_{i}$ up to $\R_n$ for which we have $\R_n \Rrightarrow \N_n$. Then $\M A, \{\N_1 \} \vdash \R_1 A$, and $\R_i A, \{\N_i\} \vdash \R_{i+1} A$ up to $\R_n , \{\N_n\} \vdash A$ are all constructable, composing into $\Gamma , \{\N_1\} , \dots , \{\N_n\} \vdash A$.
	\end{itemize}
	
	$\tLock{\M}{t}$ with $\Gamma , \{\M\} \vdash t : A$. By induction hypothesis, we end with a proof of $\Gamma \ominus \M \vdash A$ with constructable $\Gamma \vdash \Gamma$. This gives us a proof of $\Gamma \vdash \M A$ using (ModR). Composing $\Gamma \vdash \Gamma$ and $\Gamma \vdash \M A$ we construct $\Gamma \vdash \M A$.
	
	Other cases are simpler.
\end{proof}

\begin{lemma}\label{lem:shift-to-at}
	For any term $(\Gamma \ominus \M) , \Delta \vdash t : A$, there is a term $\Gamma , \{\M\} , \Delta \vdash r : A$.
\end{lemma}

\begin{proof}
	$(\Gamma \ominus \M) , \Delta \vdash t : A$ can be weakened to $\Gamma , \{\M\} , (\Gamma \ominus \M) , \Delta \vdash t : A$. We reason that for any formula $B$ from $(\Gamma \ominus \M)$ there is a lambda term $\Gamma , \{\M\} \vdash r_B : B$, and hence we can use substitutions to get a term $\Gamma , \{\M\} , \Gamma \vdash r : A$.
	Suppose $B \in (\Gamma \ominus \M)$, then either:
	\begin{itemize}
		\item $x : \N B \in \Gamma$ for some $x$, and $\N \Rrightarrow \M$. Then $\Gamma , \{\M\} \vdash \tKey{\N}{\M}{\Var{x}} : B$.
		\item $B = (\N \ominus \M)B'$ and $x : \N B' \in \Gamma$. Then $\Gamma , \{\M\} \vdash \tLock{\N \ominus \M}{\tKey{\N}{\M , \N \ominus \M}{\Var{x}}} : (\N \ominus \M)B'$.
	\end{itemize}
\end{proof}

\begin{theorem}
	If $\Gamma \vdash A$ is provable, there is a lambda term $t$ such that $\Gamma \vdash t : A$
\end{theorem}

\begin{proof}
	By induction on proofs. Most cases are standard, we shall just consider the modal case.
	
	(ModR) Proving $\Gamma \vdash \M A$, using a proof of $\Gamma \ominus \M \vdash A$. By induction hypothesis, we have a term $\Gamma \ominus \M \vdash t : A$, and we can apply Lemma \ref{lem:shift-to-at} to get $\Gamma , \{\M\} \vdash r : A$ and hence $\Gamma \vdash \tLock{\M}{s} : A$.
\end{proof}

\begin{theorem}
	If $\Gamma \vdash A$ is constructable, there is a lambda term $t$ such that $\Gamma \vdash t : A$
\end{theorem}

\begin{proof}
	By induction on $\Gamma$.
	
	If $\Gamma = \Gamma$, then $\Gamma \vdash A$ is provable, and we use the previous theorem.
	
	If $\Gamma = \Gamma , \{\M\} , \Gamma'$, with provable $\Gamma \vdash \Delta$, and constructable $\Delta \ominus \M , \Gamma' \vdash A$. We use the previous theorem and the induction hypothesis to find terms $\Gamma \vdash t_D : D$ for each $D \in \Delta$, and  $\Delta \ominus \M , \Gamma' \vdash r : A$. Applying Lemma \ref{lem:shift-to-at} to the latter, we get $\Delta , \{\M\} , \Gamma' \vdash s : A$. Substituting each $t_D$ into $s$ we get the desired term $\Gamma , \{\M\} , \Gamma' \vdash t' : A$.
\end{proof}

We can conclude that $\Gamma \vdash A$ is constructible if and only if there is a lambda term $t$ such that $\Gamma \vdash t : A$.
If in particular $\Gamma$ is flat, then $\Gamma \vdash A$ is provable in the sequent calculus if and only if there is a lambda term $t$ such that $\Gamma \vdash t : A$. Hence by decidability, there is an algorithm which determines for each sequent $\Gamma \vdash A$ with $\Gamma$ a flat context, whether a lambda term exists. This can be extended to any modal context as follows.

We define the map sending pairs $(\Gamma \vdash A)$ of modal context and formulas to a formula $\langle \Gamma , A \rangle$ as follows:
\begin{itemize}
	\item $\langle \vdash A \rangle := A$,
	\item $\langle (\Gamma , x : B) \vdash A \rangle := \langle \Gamma \vdash B \Rightarrow A \rangle$,
	\item $\langle (\Gamma , \{\M\}) \vdash A \rangle :=  \langle \Gamma \vdash \M A \rangle$.
\end{itemize}
By induction, $\langle (x : B , \Gamma) \vdash A \rangle = B \Rightarrow \langle \Gamma \vdash A \rangle$ and $\langle (\{\M\}, \Gamma) \vdash A \rangle = \M \langle \Gamma \vdash A \rangle$.

\begin{lemma}
	For each $\Gamma$, $\Delta$ and $A$, the sequence $\Gamma , \Delta \vdash A$ has a lambda term if and only if $\Gamma \vdash \langle \Delta \vdash A \rangle$ has a lambda term.
\end{lemma}

\begin{proof}
	We prove both directions of the implication by induction on the length of $\Delta$. The base case is simple, since $\langle \vdash A \rangle = A$.
	First from left to right:
	\begin{itemize}
		\item If $\Delta = \Delta' , x : B$ and $\Gamma , \Delta', x : B \vdash t : A$. Then $\Gamma , \Delta' \vdash \lambda x : B.t : B \Rightarrow A$ which by induction hypothesis gives us a term $\Gamma \vdash t' : \langle \Delta' \vdash B \Rightarrow A \rangle$, where $\langle \Delta' \vdash B \Rightarrow A \rangle = \langle \Delta', x : B \vdash A \rangle$.
		\item If $\Delta = \Delta' , \{\M\}$ and $\Gamma , \Delta , \{\M\} \vdash t : A$. Then $\Gamma , \Delta' \vdash \tLock{\M}{t} : \M A$ which by induction hypothesis gives us a term $\Gamma \vdash t' : \langle \Delta' \vdash \M A \rangle$, where $\langle \Delta' \vdash \M A \rangle = \langle \Delta', \{\M\} \vdash A \rangle$.
	\end{itemize}
	The converse:
	\begin{itemize}
		\item If $\Delta = x : B , \Delta'$ and $\Gamma \vdash t : \langle x : B , \Delta' \vdash A \rangle$, then since $\langle x : B , \Delta' \vdash A \rangle = B \Rightarrow \langle  \Delta' \vdash A \rangle$ we have $\Gamma , x : B \vdash t \cdot \texttt{var}(x) : \langle  \Delta' \vdash A \rangle$. This gives us by induction hypothesis a term $\Gamma , x : B , \Delta' \vdash t' : A$ as desired.
		\item If $\Delta = \{\M\} , \Delta'$ and $\Gamma \vdash t : \langle \{\M\} , \Delta' \vdash A \rangle$, then since $\langle \{\M\} , \Delta' \vdash A \rangle = \M \langle  \Delta' \vdash A \rangle$ we have $\Gamma , \{\M\} \vdash \tKey{\M}{\M}{t} : \langle  \Delta' \vdash A \rangle$. This gives us by induction hypothesis a term $\Gamma , \{\M\} , \Delta' \vdash t' : A$ as desired.
	\end{itemize}
\end{proof}

We conclude that $\Gamma \vdash A$ has a lambda term if and only if $\vdash \langle \Gamma \vdash A \rangle$ has a lambda term if and only if it has a proof in the decidable sequent calculus.
Hence it is decidable whether $\Gamma \vdash A$ has a lambda term.

\section{Additional Proofs}

\begin{proof}[Proof of Lemma \ref{lem:Nfac}]
	If $\N_S(a,b,d)$ then $(a,\alpha,b,\beta,d) \in S$ for some possibly empty $\alpha,\beta \in \Ag^*$. If $\N_S(b,c,d)$ then $(b,\gamma,c,\delta,d) \in S$ for some $\gamma , \delta \in \Ag^*$. By property 2 of a forwarding network, we can replace $\beta$ in $(a,\alpha,b,\beta,d)$ with $(\gamma,c,\delta)$, creating $(a,\alpha,b,\gamma,c,\delta,d) \in S$. This shows that $\N_S(a,c,d)$ as desired. Moreover, using property 1 of forwarding networks, we can show that $(b,\gamma,c,\delta,d) \in S$ as well, and hence $\N_S(b,c,d)$.
	The second property has a similar proof.
\end{proof}

\begin{proof}[Proof of Lemma \ref{lem:Jaxiom}]
	First note that if $\M$ and $\N$ satisfy axiom K and necessity, then $\M \N$ satisfies axiom K and necessity, since firstly $\vdash A$ implies $\vdash \N A$ implies $\vdash \M \N A$. 
	Secondly, $\vdash \N A \wedge \N (A \Rightarrow B) \Rightarrow \N B$ so $\M \N A \wedge \M \N (A \Rightarrow B) \Rightarrow \M \N B$ by necessity and axiom K for $\M$.
	We conclude that $\I_\alpha$ satisfies axiom K and necessity for any appropriate $\alpha$.
	
	Proving Axiom K, suppose $\J_{a \ot b}A$ and $\J_{a \ot b}(A \Rightarrow B)$, we want to show that $\J_{a \ot b}B$.
	Let $\gamma \in \Ag^*$ such that $(a,\gamma,b) \in S$, then $\I_{a \ot \gamma \ot b}A$ and $\I_{a \ot \gamma \ot b}(A \Rightarrow B)$, hence by axiom K on $\I_{a \ot \gamma \ot b}$ we get $\I_{a \ot \gamma \ot b}B$. This is for all relevant $\gamma$, hence $\J_{a \ot b}B$. 
	Proving necessity, if $\vdash A$, then $\vdash \I_{a \ot \gamma \ot b}A$ for all $\gamma$ such that $(a,\gamma,b) \in S$, hence $\vdash \J_{a \ot b}A$.
	
	For $\J_{a \ot b}A \Rightarrow \B_a\J_{a \ot b}A$ and $\J_{a \ot b}A \Rightarrow \J_{a \ot b}\B_bA$, simply note: 1) that $\I_{a \ot \gamma \ot b}$ starts with some $\I_{a \ot a'}$ for which $\I_{a \ot a'}B \Rightarrow \B_a \I_{a \ot a'}B$ for any $B$, hence $\I_{a \ot \gamma \ot b}A \Rightarrow \B_a \I_{a \ot \gamma \ot b}A$.
	2) $\I_{a \ot \gamma \ot b}$ ends with some $\I_{b' \ot b}$ for which $\I_{b' \ot b}A \Rightarrow \I_{b' \ot b}\B_bA$. Using axiom K on the rest of $\I_{a \ot \gamma \ot b}$ (which may be empty), we get $\I_{a \ot \gamma \ot b}A \Rightarrow \I_{a \ot \gamma \ot b}\B_bA$.
\end{proof}

\begin{lemma}
	$S_\G$, the set of shortest paths on a finite graph is a forwarding network.
\end{lemma}

\begin{proof}
	Note first that a shortest path between two vertices is given by a non-empty list of distinct vertices. The shortest path from $a$ to $a$ is simply given by $(a)$. We prove the three additional properties:
	\begin{enumerate}
		\item Suppose given a shortest path $(a, b_1\dots,b_n,c)$ from $c$ to $a$, then $(a, b_1\dots,b_n)$ has to be a shortest path from $b_n$ to $a$; if not, there would be a shorter path from $c$ to $a$.
		Hence $(a, b_1\dots,b_n) \in S_\G$, and by a similar argument, $ (b_1\dots,b_n,c) \in S$.
		\item Suppose $(a_1,\dots,a_n,b,c_1,\dots,c_m,d,e_1,\dots,e_k)$ is a shortest path, we leave ambiguous which elements are the source and target, since $n$ and/or $k$ could be zero. This is a path of length $n+m+k+1$.  We know that $(b,c_1,\dots,c_m,d)$ is the shortest path from $d$ to $b$. Taking some other shortest path $(b,c'_1,\dots,c'_p,d)$ from $d$ to $b$, we note that it has to have the same length, hence $p = m$. Hence $(a_1,\dots,a_n,b,c'_1,\dots,c'_p,d,e_1,\dots,e_k)$ is also a path of length $n+m+k+1$. Note in particular that $(a_1,\dots,a_n,b,c'_1,\dots,c'_p,d,e_1,\dots,e_k)$ cannot have repeats, since otherwise we could cut out the cycle and create a shorter path, contradicting that that the shortest path in of length $n+m+k+1$. We conclude that $(a_1,\dots,a_n,b,c'_1,\dots,c'_p,d,e_1,\dots,e_k)$ is also a shortest path, and hence in $S_\G$.
	\end{enumerate}
\end{proof}

\begin{lemma}
	If $\G$ is acyclic, $S'_\G$ is a forwarding network.
\end{lemma}

\begin{proof}
	This works since paths are closed under subpaths (property 1), and the result of replacing subpaths by alternative subpaths (property 2) cannot create a path without repeats.
\end{proof}

\end{document}